\documentclass[11pt,a4paper]{article}%{scrartcl}

\usepackage{tikz}
\usetikzlibrary{quantikz2}   
\usepackage{authblk}
\usepackage[left=1in, right=1in, top=1in, bottom=1in]{geometry}
\usepackage{amsthm}
\usepackage{amsfonts}
\usepackage{mathrsfs}
\usepackage[pdftex]{hyperref}
\usepackage{amssymb}
\usepackage{svg}
\usepackage[normalem]{ulem}
\usepackage{amsmath}
\usepackage{cleveref}
\usepackage{xcolor}
\usepackage{graphicx}
\usepackage{xparse}
\usepackage{subcaption}
\PassOptionsToPackage{dvipdfmx}{graphicx}
\usepackage[utf8]{inputenc}
\usepackage[T1]{fontenc}
\usepackage[makeroom]{cancel}

\usepackage[toc,page]{appendix}

\appto\appendix{\addtocontents{toc}{\protect\setcounter{tocdepth}{1}}}

\appto\listoffigures{\addtocontents{lof}{\protect\setcounter{tocdepth}{1}}}
\appto\listoftables{\addtocontents{lot}{\protect\setcounter{tocdepth}{1}}}

\numberwithin{equation}{section}

\usepackage{macros}

\title{Classification of locality preserving symmetries on spin chains}

\author[1]{Alex Bols \thanks{email: \href{abols01@phys.ethz.ch}{abols01@phys.ethz.ch}}}
\author[2]{Wojciech De Roeck \thanks{email: \href{wojciech.deroeck@kuleuven.be} {wojciech.deroeck@kuleuven.be}}}
\author[2]{Michiel De Wilde \thanks{email: \href{michiel.dewilde@student.kuleuven.be}{michiel.dewilde@student.kuleuven.be}}}
\author[2]{Bruno de O. Carvalho \thanks{email: \href{bruno.oliveira@kuleuven.be}{bruno.oliveira@kuleuven.be}}}

\affil[1]{Institute for Theoretical Physics, ETH Z{\"u}rich}
\affil[2]{Instituut voor Theoretische Fysica, KU Leuven}

\begin{document}

\date{\today}
\maketitle

\begin{abstract}
We consider the action of a finite group $G$ by locality preserving automorphisms (quantum cellular automata) on quantum spin chains.  We refer to such group actions as ``symmetries''. 
The natural notion of equivalence for such symmetries is \emph{stable equivalence}, which allows for stacking with factorized group actions. Stacking also endows the set of equivalence classes with a group structure.   
We prove that the anomaly of such symmetries provides an isomorphism between the group of stable equivalence classes of symmetries with the cohomology group $H^3(G,U(1))$, consistent with previous conjectures. This amounts to a complete classification of locality preserving symmetries on spin chains.
\end{abstract}

%% TABLE OF CONTENT
\tableofcontents

%% SECTIONS
%%%%%%%%%%%%%%%%%%%%%%%%%%%%%%%%%%%%%%%%%%%%%%%%%%%
%%%%%%%%%%%%%%%%%%%%%%%%%%%%%%%%%%%%%%%%%%%%%%%%%%%
%%%%%%%%%%%%%%%%%%%%%%%%%%%%%%%%%%%%%%%%%%%%%%%%%%%
\section{Introduction} \label{sec:introduction}

Dynamics in many-body quantum physics is typically generated by a local Hamiltonian, and  therefore, due to Lieb-Robinson bounds \cite{lieb1972finite},  it  \emph{preserves locality}. 
Such locally generated evolutions may be thought of as topologically trivial locality preserving automorphisms of the observable algebra. Indeed, they are contracted to the identity by reducing the evolution time. In contrast, the \emph{shift} on a spin chain is an example of a \emph{topologically non-trivial} locality preserving automorphism \cite{Gross_2012}. 

Non-trivial locality preserving automorphisms appear in various guises in the study of topological phases of strongly interacting quantum matter, and often serve to characterize and even classify the phases under investigation \cite{hastings2013classifying}. Examples include the appearance at stroboscopic times of the shift on the boundary of many-body localized Floquet insulators \cite{po2016chiral, else2016classification, zhang2023bulk}, the exotic symmetries appearing on the boundaries of topological matter \cite{kitaev2012models, jones2023local, molnar2022matrix,anomalies-entanglement-renormalization,gauging-edge-theories}, and the equivariant automorphisms that can entangle symmetry protected trivial (SPT) phases \cite{chen2013symmetry, zhang2023bulk,FidkowskiHaahHastings2025}. These connections have motivated a growing body of work that aims to understand the topological phases of locality preserving automorphisms, possibly in the presence of symmetry \cite{Gross_2012, freedman2020classification, freedman2022group, cirac2017matrix, gong2020classification, ranard2022converse, KapustinSopenko2025, bols2021classification,HaahFidkowskiHastings2022, zhang2023topological,Shirley2022QCA3D}. 

In this paper we study representations of finite groups $G$ by locality preserving automorphisms on spin chains \cite{rubio2024classifying, garre2023classifying}. Such representations can be regarded as the group case of categorical symmetries on spin chains \cite{feiguin2007interacting, lootens2023dualities, lootens2024dualities,JonesLim2024FusionQCA}. They also arise at the boundaries of two-dimensional SPTs, whose bulk invariant manifests itself as an $H^3(G, U(1))$-valued \emph{anomaly} of the boundary symmetry \cite{else2014classifying}. We prove the folk knowledge that this anomaly \emph{completely} classifies locality preserving symmetries on spin chains up to \emph{stable equivalence}. That is, up to conjugation by finite depth quantum circuits and stacking with factorized group actions.

The idea of defining the anomaly by considering the fusion properties of symmetry defects is well understood \cite{chen2011two, levin2012braiding, cheng2023lieb, KapustinSopenko2025, seifnashri2024lieb}. We review this construction in Section \ref{sec:anomaly} below. The bulk of our article is devoted to the non-trivial task of proving completeness of the classification. This means that, if two locality preserving symmetries have the same anomaly, then they are stably equivalent to each other. To achieve this, we introduce several new concepts to the study of locality preserving symmetries. The most notable new concepts are an obstruction to admitting `covariant right restrictions', as well as variations on a construction which first appeared in the operator algebra literature \cite{sutherland1980cohomology} (see Lemmas \ref{lem:group hom restrictions} and \ref{lem:covariant right restrictions}).

Our proof proceeds by first asking which locality preserving symmetries on spin chains admit restrictions to right half-lines that are
\begin{enumerate}[label=(\arabic*)]
    \item themselves locality preserving symmetries,
    \item covariant with respect to the full symmetry.
\end{enumerate}
The anomaly is an obstruction to (1). We introduce in addition a new obstruction to (2), called the obstruction to covariant right restrictions, which takes values in twisted group cohomology $H^2(G, U(1)[G])$. We then show that a symmetry with trivial anomaly also has trivial obstruction to covariant right restrictions, and that the existence of right restrictions that satisfy (1) and (2) simultaneously implies that symmetries with trivial anomaly can be decoupled. The solution of the classification problem then follows from the fact that the anomaly and the obstruction to covariant right restrictions are constant on stable equivalence classes of symmetries.

As a corollary to our classification result, we show that the obstruction to covariant right restrictions is given by the inverse of the \emph{slant product} of the anomaly. Having a good handle on this quantity is significant because it plays an important role in characterizing the anyon content of the gauged bulk SPT corresponding to the boundary symmetry under consideration \cite{wang2015topological, dijkgraaf1990topological}.

The paper is structured as follows. In Section \ref{sec:setup}, we introduce locality preserving symmetries on spin chains and stable equivalence between them, and state our main Theorem. We define the anomaly in Section \ref{sec:anomaly}, and state its basic properties. In Section \ref{sec:examples}, we construct for each element of $H^3(G, U(1))$ an explicit symmetry with that element as its anomaly. In Section \ref{sec:right restrictions}, we define the obstruction to covariant right restrictions and show that it vanishes for symmetries with trivial anomaly. This fact is then used in Section \ref{sec:proof of main thm} to prove the main Theorem.  Appendix \ref{app:cohomology} collects basic definitions of group cohomology. Basic properties of the anomaly and of the obstruction to covariant right restrictions are proved in Appendices \ref{app:proof of anomaly proposition} and \ref{app:proof of proposition covariant obstruction} respectively. In Appendix \ref{app:general_covariance} we show that the obstruction to covariant right restrictions is given by the slant product of the anomaly.\\

\noindent\textbf{Note : } During the preparation of this manuscript, the preprint \cite{seifnashri2025disentangling} appeared, in which similar results are obtained. In particular, the \emph{disentangler} $\mathcal W$ constructed in \cite{seifnashri2025disentangling} yields a  proof of our Proposition \ref{prop:trivial anomaly implies stably equivalent to decoupled symmetry}. 
\\

\noindent \textbf{Acknowledgments : }
W.D.R. and B.O.C. were supported by the FWO (Flemish Research Fund) grant G098919N, 
the FWO-FNRS EOS research project G0H1122N EOS 40007526 CHEQS, the KULeuven  Runners-up grant iBOF DOA/20/011, and the internal KULeuven grant C14/21/086.
%%%%%%%%%%%%%%%%%%%%%%%%%%%%%%%%%%%%%%%%%%%%%%%%%%%
%%%%%%%%%%%%%%%%%%%%%%%%%%%%%%%%%%%%%%%%%%%%%%%%%%%
%%%%%%%%%%%%%%%%%%%%%%%%%%%%%%%%%%%%%%%%%%%%%%%%%%%
\section{Setup and main result} \label{sec:setup}

%%%%%%%%%%%%%%%%%%%%%%%%%%%%%%%%%%%%%%%%%%%%%%%%%%%
%%%%%%%%%%%%%%%%%%%%%%%%%%%%%%%%%%%%%%%%%%%%%%%%%%%
\subsection{Spin chains, quantum cellular automata, and finite depth quantum circuits} \label{subsec:spin chains, QCA, and FDQC}

A spin chain $C^*$-algebra $\caA$ is defined in the standard way, that we recall now.
To any site $j \in \bbZ$, we associate an $d_j$-dimensional on-site Hilbert space $\C^{d_j}$, with associated matrix algebra $\caA_j \simeq \End(\C^{d_j})$. 
We assume that there is a $d_{\max}$ such that $d_j \leq d_{\max}$. 
The algebra $\caA_j \simeq \End(\C^{d_j})$ is equipped with its natural operator norm and $*$-operation (Hermitian adjoint of a matrix) making it into a $C^*$-algebra.  The spin chain algebra $\caA$ is the inductive limit of 
algebras $\caA_S=\otimes_{j\in S} \caA_j$, with  $S$ a finite subset of $\Z$. It comes naturally equipped with local subalgebras $\caA_X, X\subset \Z$. We refer to standard references \cite{bratteliII,simon2014statistical,naaijkens2017quantum,brunoamandaI,nachtergaele.sims.ogata.2006} for more background and details.
We will write $\caA_{\geq j}$ for $\caA_{[j, \infty)}$ and $\caA_{< j}$ for $\caA_{(-\infty, j-1]}$. We will usually refer to the quasi-local algebra $\caA$ itself as the \textit{spin chain}, it being understood that there is a fixed preferred assignment of on-site algebras $j \mapsto \caA_j \subset \caA$ to sites of $\Z$.

For any $\Gamma \subset \Z$ we write $\Gamma^{(r)} := \{ j \in \Z \, : \, \dist(j, \Gamma) \leq r \}$
 for the \textit{$r$-fattening} of $\Gamma$. A \emph{quantum cellular automaton} (QCA) on a spin chain $\caA$ is a $*$-automorphism $\al : \caA \rightarrow \caA$ for which there exists $r \geq 0$ such that $\alpha(\caA_{X}) \subset  \caA_{X^{(r)}} $ for any $X\subset \Z$.

The range of a QCA is the smallest $r$ for which this holds. The inverse of a QCA of range $r$ is also a QCA of range $r$ (\cite[Lemma 3.1]{freedman2020classification}), a fact which we will use without comment throughout the paper. The quantum cellular automata on $\caA$ form a subgroup of $\Aut(\caA)$ which we denote by $\QCA(\caA)$.

Let $\{ I_a \}_{a \in \Z}$ be a partition of $\Z$ into intervals $I_a \subset \Z$ of bounded size. Suppose we have for each $a \in \Z$ a unitary $U_a \in \caA_{I_a}$, then we can define a QCA $\beta$ by the formal infinite product
$$
\beta=\bigotimes_{a \in \Z} \Ad(U_a).
$$
This yields a well-defined automorphism, as one can first define its action on $\caA_X$ with finite $X$ and then extend by density. 
 Any QCA of this form is called a \emph{block partitioned} QCA. The intervals $I_a$ are called the blocks of the block partitioned QCA, and $\abs{I_a}$ is the size of block $I_a$. The unitaries $U_a$ are called \emph{gates}. The composition of $n$ block partitioned QCAs is called a depth $n$ quantum circuit, or simply a \emph{finite depth quantum circuit} (FDQC).

%%%%%%%%%%%%%%%%%%%%%%%%%%%%%%%%%%%%%%%%%%%%%%%%%%%
%%%%%%%%%%%%%%%%%%%%%%%%%%%%%%%%%%%%%%%%%%%%%%%%%%%
\subsection{Locality preserving symmetries} \label{subsec:locality preserving symmetries}

Let $G$ be a finite group which will be fixed throughout the paper. We write $\bar g = g^{-1}$ for the inverse of any group element $g \in G$. A \emph{locality preserving symmetry} on $\caA$ is a group homomorphism $\al : G \rightarrow \QCA(\caA)$. That is, for each $g \in G$ we have a quantum cellular automaton $\al^{(g)}$ such that $\al^{(1)} = \id$ and $\al^{(g)} \circ \al^{(h)} = \al^{(gh)}$ for all $g, h \in G$. The range of a locality preserving symmetry is the largest range of its component QCAs. We say a locality preserving symmetry $\al$ is \emph{decoupled}
iff. every $\al^{(g)}, g\in G$ is a block-partitioned QCA (as defined above) where the blocks can be chosen to be $g$-independent.  Alternatively, this means that we can write 
 formally $\al = \bigotimes_{a \in \Z} \, \al_a$ for symmetries $\al_a$ supported on the blocks $I_a$. In particular, any symmetry of range 0 is decoupled.

In the rest of this work we will refer to locality preserving symmetries simply as symmetries. If we want to specify the group $G$ then we speak of $G$-symmetries. We denote the set of all $G$-symmetries on arbitrary spin chains by $\Sym_{G}$.

%%%%%%%%%%%%%%%%%%%%%%%%%%%%%%%%%%%%%%%%%%%%%%%%%%%
%%%%%%%%%%%%%%%%%%%%%%%%%%%%%%%%%%%%%%%%%%%%%%%%%%%
\subsection{Equivalence and stable equivalence} \label{subsec:stable equivalence}

Two symmetries $\al$ and $\al'$ are \textit{equivalent} if there is a FDQC $\gamma$ such that $\al'^{(g)} = \gamma^{-1} \circ \al^{(g)} \circ \gamma $ holds for all $g\in G$. In that case we write $\al' \sim_0 \al$.

The \emph{stack} of two spin chains $\caA$ and $\caB$ is the spin chain $\caA \otimes \caB$ with on-site algebras $(\caA \otimes \caB)_j = \caA_j \otimes \caB_j$ for all $j \in \Z$. If $\al$ and $\beta$ are symmetries on the spin chains $\caA$ and $\caB$ respectively, then we can \emph{stack} them to obtain the symmetry $\al \otimes \beta$ on $\caA \otimes \caB$ with components $(\al \otimes \beta)^{(g)} = \al^{(g)} \otimes \beta^{(g)}$ for all $g \in G$.

Two symmetries $\al$ and $\al'$ are \textit{stably equivalent}, denoted by $\al \sim \al'$, if there exists symmetries $\beta $ and $\beta'$ of range $0$ such that $\al \otimes \beta \sim_0 \al' \otimes \beta'$. Stable equivalence is an equivalence relation on $\Sym_G$, and $(\Sym_G / \sim)$ is an abelian monoid with multiplication induced by stacking. (We will show later that it is in fact a group, \ie there are inverses.) 
It is easy to check that any decoupled symmetry is stably equivalent to a symmetry of range zero. This implies also that $\alpha,\alpha'$ are stably equivalent whenever there exist decoupled symmetries $\beta,\beta'$ such that $\al \otimes \beta \sim_0 \al' \otimes \beta'$. This fact will be used throughout the paper without further mention.

%%%%%%%%%%%%%%%%%%%%%%%%%%%%%%%%%%%%%%%%%%%%%%%%%%%
%%%%%%%%%%%%%%%%%%%%%%%%%%%%%%%%%%%%%%%%%%%%%%%%%%%
\subsection{Main result} \label{subsec:main result}

\begin{theorem} \label{thm:classification}
    The monoid $(\Sym_G/\sim)$ is in fact a group. There is a map $\Omega : \Sym_G \rightarrow H^3(G, U(1))$ which assigns to each symmetry $\al$ a 3-cohomology class, which we will call its anomaly, and which lifts to an isomorphism of groups $(\Sym_G / \sim) \cong H^3(G, U(1))$.

    In particular, two $G$-symmetries $\al$ and $\beta$ are stably equivalent if, and only if, their anomalies are equal: $$\al \sim \beta \, \iff \, \Omega(\al) = \Omega(\beta).$$ Moreover, for each $[\omega] \in H^3(G, U(1))$ there exists a symmetry whose anomaly is $[\omega]$.
\end{theorem}

This theorem is proven at the end of Section \ref{sec:proof of main thm}.

\begin{remark}
   Appendix B.1 of \cite{zhang_long-range_2024} describes a $\Z_2$ symmetry $\al$ with trivial anomaly $\Omega(\al) = [1]$ which is nevertheless not equivalent to a decoupled symmetry. This shows that the notion of stable equivalence is indeed necessary for the classification by the anomaly to hold.
\end{remark}

%%%%%%%%%%%%%%%%%%%%%%%%%%%%%%%%%%%%%%%%%%%%%%%%%%%
%%%%%%%%%%%%%%%%%%%%%%%%%%%%%%%%%%%%%%%%%%%%%%%%%%%
%%%%%%%%%%%%%%%%%%%%%%%%%%%%%%%%%%%%%%%%%%%%%%%%%%%
\section{The anomaly of a locality preserving symmetry} \label{sec:anomaly}

The idea behind the definition of the anomaly presented here goes back to \cite{else2014classifying}. In order to define the anomaly we first note that the component QCAs of any locality preserving symmetry are finite depth quantum circuits \cite{zhang2023topological}.

\begin{lemma} \label{lem:LPSs are FDQCs}
    Let $\al : G \rightarrow \QCA(\caA)$ be a symmetry of range $R$ on a spin chain $\caA$. Then each $\al^{(g)}$ can be written as a depth two quantum circuit whose blocks all have size at most $2R$.
\end{lemma}

\begin{proof}
    For each $g \in G$ we have a QCA $\al^{(g)}$ on the spin chain $\caA$. To any such QCA one can assign its $\Q$-valued GNVW index $\ind( \al^{(g)} ) \in \Q$, see \cite{Gross_2012}. Since $G$ is a finite group, $g$ has finite order. \ie there is an $n$ such that $g^n = 1$. Since the GNVW index is multiplicative under composition of QCAs and $\ind(\id) = 1$, this implies that $\ind( \al^{(g)} )^n = \ind( \al^{(g^n)} ) = \ind( \al^{(1)} ) = \ind( \id ) = 1$ and therefore $\al^{(g)}$ has trivial GNVW index. The claim now follows from~\cite[Theorem 9]{Gross_2012}.
\end{proof}

Let $\al$ be a symmetry of range $R$. A \emph{right restriction} $\al_{\geq j}$ of $\al$ at $j \in \Z$ with \emph{defect size} $L$ is a family of automorphisms $\al_{\geq j}^{(g)}$ such that for any $g\in G$ 
$$
\al_{\geq j}^{(g)} |_{\caA_{\geq (j + L)}} = \al^{(g)}|_{\caA_{\geq (j + L)}} \, \text{ and } \,\, \al_{\geq j}^{(g)}|_{\caA_{< j-L}} = \id_{\caA_{< j-L}}.
$$
It follows immediately from Lemma \ref{lem:LPSs are FDQCs} that any symmetry of range $R$ admits right restrictions at all sites with defect size $2R$.

Given a right restriction $\al_{\geq j}$ of defect size $L$, there are local unitaries $\Fusion_{j}(g, h) \in \caA_{[j - L, j + L + R)}$, called \emph{fusion operators} associated to $\al_{\geq j}$, such that
$$
\al_{\geq j}^{(g)} \circ \al_{\geq j}^{(h)} = \Ad[\Fusion_{j}(g, h)] \circ \al_{\geq j}^{(gh)}.
$$
These unitaries are uniquely determined by this equation up to phase. They capture the failure of $g \mapsto \al_{\geq j}^{(g)}$ to be a group homomorphism.

Using associativity to compute $\al_{\geq j}^{(f)} \circ \al_{\geq j}^{(g)} \circ \al_{\geq j}^{(h)}$ in two different ways one obtains
$$
\Ad \left[ \Fusion_{j}(f, g) \Fusion_{j}(fg, h) \right] \circ \al_{\geq j}^{(fgh)} = \Ad \left[ \al_{\geq j}^{(f)} \big( \Fusion_{j}(g, h) \big) \Fusion_{j}(f, gh) \right] \circ \al_{\geq j}^{(fgh)}.
$$
It follows that there are phases $\omega_{j}(f, g, h) \in U(1)$ such that
\begin{equation} \label{eq-fusion.operators}
\Fusion_{j}(f, g) \Fusion_{j}(fg, h) = \omega_{j}(f, g, h) \times \al_{\geq j}^{(f)} \big( \Fusion_{j}(g, h) \big) \Fusion_{j}(f, gh)
\end{equation}
for all $f, g, h \in G$.

\begin{proposition} \label{prop:anomaly}
    The map $\omega_{j} : G^3 \rightarrow U(1)$ is a 3-cocycle,
    $$ 1 = \frac{\omega_{j}(g, h, k) \omega_{j}(f, gh, k) \omega_{j}(f, g, h)}{\omega_{j}(fg, h, k) \omega_{j}(f, g, hk)},$$
    and the corresponding group cohomology class $[\omega_{j}] \in H^3(G, U(1))$ depends only on the symmetry $\al$, i.e. the cohomology class is independent of the site $j$ and the choice of right restriction $\al_{\geq j}$. We thus obtain a well defined map $$\ano : \Sym_G \rightarrow H^3(G, U(1))$$ which we call the anomaly. If $\ano(\al) = [1]$ is the identity element of $H^3(G, U(1))$, then we say that $\al$ has trivial anomaly.

    Moreover, for symmetries $\al$ and $\beta$ of range $R$ we have
    \begin{enumerate}
        \item \label{propitem:anomaly decoupled implies trivial} If $\al$ is decoupled then $\ano(\al) = [1]$ is the identity element of $H^3(G, U(1))$. 
        \item \label{propitem:anomaly locally computable} The anomaly is locally computable: If $\al$ and $\beta$ act on the same spin chain $\caA$ and there is an interval $I$ of length $8R$ such that $\al|_{\caA_I} = \beta|_{\caA_I}$ then $\ano(\al) = \ano(\beta)$. 
        \item \label{propitem:anomaly multiplicative} The anomaly is multiplicative under stacking: $\ano(\al \otimes \beta) = \ano(\al) \cdot \ano(\beta)$.
        \item \label{propitem:anomaly constant on stable equivalence classes} The anomaly is constant on stable equivalence classes: $\al \sim \beta \implies \ano(\al) = \ano(\beta)$.  
    \end{enumerate}
    In particular, the anomaly lifts to a homomorphism of monoids $\ano : (\Sym_G / \sim) \rightarrow H^3(G, U(1))$.
\end{proposition}

The proof can be found in Appendix \ref{app:proof of anomaly proposition}.

\begin{remark} \label{rem:could use left restrictions}
        We could use left restrictions instead of right restrictions to give an alternative anomaly $\ano_L(\al)$. Then one can check that $\ano_L(\al) = \ano(\al)^{-1}$.
\end{remark}

%%%%%%%%%%%%%%%%%%%%%%%%%%%%%%%%%%%%%%%%%%%%%%%%%%%
%%%%%%%%%%%%%%%%%%%%%%%%%%%%%%%%%%%%%%%%%%%%%%%%%%%
%%%%%%%%%%%%%%%%%%%%%%%%%%%%%%%%%%%%%%%%%%%%%%%%%%%
\section{Examples} \label{sec:examples}

% \mich{Should here be references to fusion chains?}
Let $G$ be a finite group and $\omega : G^3 \rightarrow U(1)$ a 3-cocycle. We construct a symmetry $\al$ with anomaly $[\omega] \in H^3(G, U(1))$.

Consider the spin chain with on-site algebras $\caA_x \simeq \End \big( \C^{\abs{G}} \big)$. Define unitaries $V_{j, j+1}^{(g)} \in \caA_{\{j, j+1\}}$ by
$$
V_{j, j+1}^{(g)} |g_j, g_{j+1} \rangle = \omega(g, g_{j+1}, \bar g_{j+1} g_{j}) |g_j,  g_{j+1} \rangle.
$$
Note that the $V_{j, j+1}^{(g)}$ commute with each other for all $j \in \Z$ and for all $g \in G$.

Define $\al^{(g)}$ as the composition $\al_3^{(g)} \circ \al_2^{(g)} \circ \al_1^{(g)}$ of three block partitioned QCAs. The blocks of $\al_1^{(g)}$ are neighbouring pairs of sites $\{2a, 2a+1\}$ and the corresponding gates are $V^{(g)}_{2a, 2a+1}$. Similarly, $\al_2^{(g)}$ has blocks $\{2a-1, 2a\}$ and corresponding gates $V_{2a-1, 2a}^{(g)}$. Finally, the block partitioned QCA $\al_3^{(g)}$ has the singletons $\{a\}$ as blocks with the left action $L^{(g)} | h \rangle = | gh \rangle$ as gates. See Figure \ref{fig:symmetry example}.

%-----------------------------------------
%----FIGURE 1-----------------------------
\begin{figure} [ht]
\centering
\begin{center}
\resizebox{.6\linewidth}{!}{
\rotatebox{-90}{
\begin{quantikz}
%    & \gate{\vdots} & \gate{\vdots} & \gate{\vdots} & \\
    & \gate{\vdots} & \gate[label style={black,rotate=90},wires = 2][45pt][0]{V_{9,10}^{(g)}} &  \gate{\vdots}  & \qw \\
    &  \gate[label style={rotate=90}][15pt][25pt]{L^{(g)}} & & \gate[label style={black,rotate=90},wires = 2][45pt][0]{V_{8,9}^{(g)}} & \qw \\
    & \gate[label style={rotate=90}][15pt][25pt]{L^{(g)}} & \gate[label style={black,rotate=90},wires = 2][45pt][0]{V_{7,8}^{(g)}} & &\qw \\
    & \gate[label style={rotate=90}][15pt][25pt]{L^{(g)}} & & \gate[label style={black,rotate=90},wires = 2][45pt][0]{V_{6,7}^{(g)}} &  \qw \\
    & \gate[label style={rotate=90}][15pt][25pt]{L^{(g)}} & \gate[label style={black,rotate=90},wires = 2][45pt][0]{V_{5,6}^{(g)}} & &\qw \\
    & \gate[label style={rotate=90}][15pt][25pt]{L^{(g)}} & & \gate[label style={black,rotate=90},wires = 2][45pt][0]{V_{4,5}^{(g)}}& \qw \\
    & \gate[label style={rotate=90}][15pt][25pt]{L^{(g)}} & \gate[label style={black,rotate=90},wires = 2][45pt][0]{V_{3,4}^{(g)}} & &\qw \\
    & \gate[label style={rotate=90}][15pt][25pt]{L^{(g)}} & & \gate[label style={black,rotate=90},wires = 2][45pt][0]{V_{2,3}^{(g)}}&\qw  \\
    & \gate[label style={rotate=90}][15pt][25pt]{L^{(g)}} & \gate[label style={black,rotate=90},wires = 2][45pt][0]{V_{1,2}^{(g)}} & & \qw \\
    & \gate{\vdots} & & \gate{\vdots} &\qw 
\end{quantikz}}}
\end{center}

\caption{The FDQC defining $\alpha^{(g)}$.}
\label{fig:symmetry example}
\end{figure}
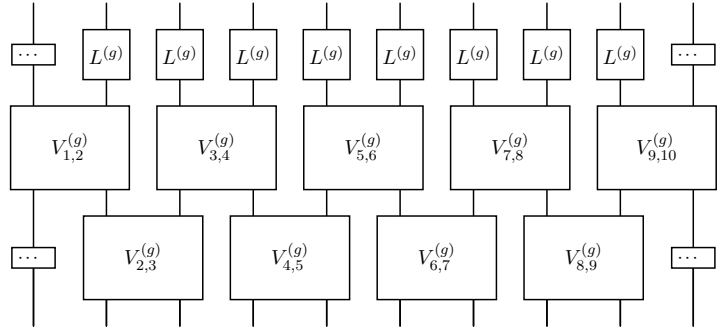

%-----------------------------------------
%----END OF FIGURE 1----------------------

Let $I = [a, b] \subset \Z$ be a finite interval and let $\al_I^{(g)}$ be the FDQC obtained from $\al^{(g)}$ by only retaining the gates that are supported on $I$, see Figure \ref{fig:symmetry example restricted}. The product of the finite number of gates of $\al_I^{(g)}$ then defines a unitary  $U_I^{(g)}$ so that $\al_I^{(g)} = \Ad[ U_I^{(g)} ]$. Note that $\al^{(g)} = \lim_{a \uparrow \infty} \al_{[-a, a]}^{(g)}$ in the strong topology.

%-----------------------------------------
%----FIGURE 2-----------------------------
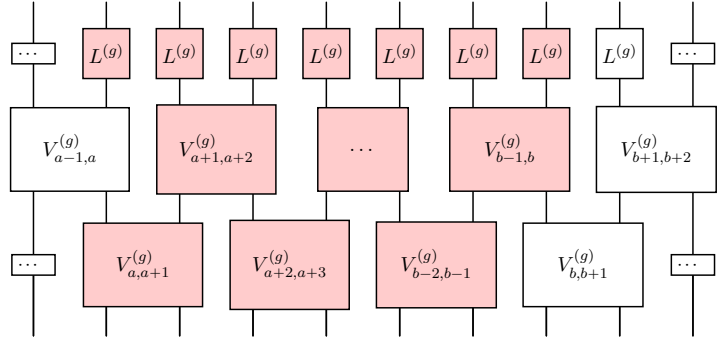
\begin{figure} [ht]
\centering
\begin{center}
\resizebox{.6\linewidth}{!}{
\rotatebox{-90}{
\begin{quantikz}
%    & \gate{\vdots} & \gate{\vdots} & \gate{\vdots} & \\
    & \gate{\vdots} & \gate[label style={black,rotate=90},wires = 2][45pt][0]{V_{b+1,b+2}^{(g)}} &  \gate{\vdots}  & \qw \\
    &  \gate[label style={rotate=90}][15pt][25pt]{L^{(g)}} & & \gate[label style={black,rotate=90},wires = 2][45pt][0]{V_{b,b+1}^{(g)}} & \qw \\
    & \gate[label style={rotate=90}, style = {fill = red!20}][15pt][25pt]{L^{(g)}} & \gate[label style={black,rotate=90},wires = 2, style = {fill = red!20}][45pt][0]{V_{b-1,b}^{(g)}} & &\qw \\
    & \gate[label style={rotate=90}, style = {fill = red!20}][15pt][25pt]{L^{(g)}} & & \gate[label style={black,rotate=90},wires = 2, style = {fill = red!20}][45pt][0]{V_{b-2,b-1}^{(g)}} &  \qw \\
    & \gate[label style={rotate=90}, style = {fill = red!20}][15pt][25pt]{L^{(g)}} & \gate[label style={black,rotate=90},wires = 2, style = {fill = red!20}][45pt][0]{\dots} & &\qw \\
    & \gate[label style={rotate=90}, style = {fill = red!20}][15pt][25pt]{L^{(g)}} & & \gate[label style={black,rotate=90},wires = 2, style = {fill = red!20}][45pt][0]{V_{a+2,a+3}^{(g)}}& \qw \\
    & \gate[label style={rotate=90}, style = {fill = red!20}][15pt][25pt]{L^{(g)}} & \gate[label style={black,rotate=90},wires = 2, style = {fill = red!20}][45pt][0]{V_{a+1,a+2}^{(g)}} & &\qw \\
    & \gate[label style={rotate=90}, style = {fill = red!20}][15pt][25pt]{L^{(g)}} & & \gate[label style={black,rotate=90},wires = 2, style = {fill = red!20}][45pt][0]{V_{a,a+1}^{(g)}}&\qw  \\
    & \gate[label style={rotate=90}, style = {fill = red!20}][15pt][25pt]{L^{(g)}} & \gate[label style={black,rotate=90},wires = 2][45pt][0]{V_{a-1,a}^{(g)}} & & \qw \\
    & \gate{\vdots} & & \gate{\vdots} &\qw 
\end{quantikz}}}
\end{center}

\caption{The FDQC (in red) defining $\alpha_I^{(g)}$ for $I = [a, b]$.}
\label{fig:symmetry example restricted}
\end{figure}

%-----------------------------------------
%----END OF FIGURE 2----------------------

\begin{lemma} \label{lem:example fusion operators}
    If $I = [a, b]$ then
    $$
    U_I^{(g)} U_I^{(h)} U_I^{(gh)*} = \Fusion_a(g, h) \times \Fusion_b(g, h)^*
    $$
    with
    $$
    \Fusion_j(g, h) = \sum_{g_j \in G} \omega(g, h, \bar h \bar g g_j) \, | g_j \rangle \langle g_j |.
    $$
\end{lemma}

\begin{proof}
    Let us act on a product state $|g_a \, g_{a+1} \cdots g_b \rangle$. Let $f_j = \bar g_{j} g_{j-1}$ for all $j \in \{a+1, \cdots, b\}$. Then
    $$
    U_I^{(h)} |g_a \, g_{a+1} \cdots g_b \rangle = \prod_{j = a+1}^b \omega(h, g_j, f_j) \times |(h g_a) \, (h g_{a+1}) \cdots (h g_b) \rangle
    $$
    so
    $$
    U_I^{(g)} U_I^{(h)} |g_a \, \cdots g_b \rangle = \prod_{j = a+1}^b \omega(h, g_j, f_j) \omega(g, h g_j, f_j) \times |(gh g_a) \, \cdots (gh g_b) \rangle
    $$
    and
    \begin{align*}
        U_I^{(g)} U_I^{(h)} U_I^{(gh)*} | (gh g_a) \cdots (gh g_b) \rangle &= \prod_{j = a+1}^b \frac{\omega(h, g_j, f_j) \omega(g, h g_j, f_j)}{\omega(gh, g_j, f_j)} \times |(gh g_a) \, \cdots (gh g_b) \rangle \\
        &= \prod_{j = a+1}^b \frac{\omega(g, h, g_{j-1})}{\omega(g, h, g_j)} \times |(gh g_a) \, \cdots (gh g_b) \rangle \\
        &= \frac{\omega(g, h, g_a)}{\omega(g, h, g_b )} \times |(gh g_a) \, \cdots (gh g_b) \rangle \\
        &= \Fusion_a(g, h) \Fusion_b(g, h)^* \times |(gh g_a) \, \cdots (gh g_b) \rangle
    \end{align*}
    where we recalled that $f_j = \bar g_j g_{j-1}$ and we used the cocycle relation (Proposition \ref{prop:anomaly}) in the second step.
\end{proof}

It follows immediately from this lemma that $\al^{(g)} \al^{(h)} = \al^{(gh)}$ for all $g, h \in G$, so that $g \mapsto \al^{(g)}$ is indeed a group homomorphism. Let $\al_{\geq}^{(g)} := \lim_{b \uparrow \infty} \al_{[0, b]}^{(g)}$. Then $\al_{\geq}$ is a right restriction of $\al$ and
$$
\al_{\geq}^{(g)} \al_{\geq}^{(h)} = \Ad[ \Fusion_0(g, h) ] \al_{\geq}^{(gh)}.
$$

We can now compute the anomaly of $\al$. Note that the fusion operators $\Fusion_j(g, h)$ commute with all the gates $V_{j, j+1}^{(f)}$, indeed, all these operators are diagonal in the `group basis'. We therefore find
\begin{align*}
    \al_{\geq}^{(f)} \big( \Fusion_0(g, h) \big) = \sum_{g_0 \in G} \, \omega(g, h, \bar h \bar g g_0)  \, | f g_0 \rangle \langle f g_0 | = \sum_{g_0 \in G} \, \omega(g, h, \bar h \bar g \bar f g_0) \, | g_0 \rangle \langle g_0 |
\end{align*}
hence
$$
 \al_{\geq}^{(f)} \big( \Fusion_0(g, h) \big) \Fusion_0(f, gh) = \sum_{g_0 \in G} \omega(g, h, \bar h \bar g \bar f g_0) \omega(f, gh, \, \bar h \bar g \bar f g_0) \, | g_0 \rangle \langle g_0 |
$$
Comparing this to
$$
\Fusion_0(f, g) \Fusion_0(fg, h) = \sum_{g_0 \in G} \omega(f, g, \bar g \bar f g_0) \omega(fg, h, \bar h \bar g \bar f g_0) |g_0 \rangle \langle g_0 | 
$$
and using the cocycle relation yields
$$
\Fusion_0(f, g) \Fusion_0(fg, h) = \omega(f, g, h) \times \al_{\geq}^{(f)} \big( \Fusion_0(g, h) \big) \Fusion_0(f, gh),
$$
showing that $\al$ indeed has anomaly $[\omega]$.
%%%%%%%%%%%%%%%%%%%%%%%%%%%%%%%%%%%%%%%%%%%%%%%%%%%
%%%%%%%%%%%%%%%%%%%%%%%%%%%%%%%%%%%%%%%%%%%%%%%%%%%
%%%%%%%%%%%%%%%%%%%%%%%%%%%%%%%%%%%%%%%%%%%%%%%%%%%
\section{Right restrictions and covariance} \label{sec:right restrictions}

%%%%%%%%%%%%%%%%%%%%%%%%%%%%%%%%%%%%%%%%%%%%%%%%%%%
%%%%%%%%%%%%%%%%%%%%%%%%%%%%%%%%%%%%%%%%%%%%%%%%%%%
\subsection{Right restrictions that are group homomorphisms}

The following Lemma says that any symmetry with trivial anomaly admits right restrictions that are group homomorphisms, at least after stacking with a local degree of freedom. The proof is closely analogous to~\cite[Theorem 3.1.6]{sutherland1980cohomology}. \par 

A \emph{local degree of freedom} at site $j$ is a spin chain where the on-site algebras $\caA_i \simeq \C$ are trivial for $i\neq j$. For convenience, also the only non-trivial on-site algebra $\caA_j$ is sometimes also called the local degree of freedom. The procedure of stacking a local degree of freedom $\caA_j'$ with a spin chain $\caA$ results in a spin chain that is isomorphic to $\caA \otimes \caA'_j$. A simultaneous stacking of a countable number of (uniformly upper bounded) local degrees of freedom at different sites with a spin chain $\caA$ still produces a spin chain with uniformly upper bounded on-site dimensions.

\begin{lemma} \label{lem:group hom restrictions}
Let $\al$ be a symmetry of range $R$ with trivial anomaly. 
For any site $j \in \Z$ we can stack the on-site algebra $\caA_j$ by a local degree of freedom $\End(\C^{\abs{G}})$, obtaining an enlarged spin chain $\widetilde \caA$. Then the composed symmetry $\tilde \al = \al \otimes \id$ on $\widetilde \caA$ admits a right  restriction $\tilde \beta_{\geq j}$ at site $j$ of defect size $5R$ such that $g \mapsto \tilde \beta_{\geq j}^{(g)}$ is a group homomorphism.
\end{lemma}

\begin{proof}
    Let $\al_{\geq j}$ be a right restriction of $\alpha$ at $j \in \Z$ with defect size $2R$. We drop $j$ from the notation in the remainder of this proof. Let $\Fusion(g, h) \in \caA_{[j-2R, j+3R)}$ be the fusion operators associated to this right restriction. Since $\al$ has trivial anomaly, by Eq. \eqref{eq-fusion.operators} we can choose the phases of the fusion operators such that
    \begin{equation} \label{eq:trivial anomaly relation}
        \Fusion(f, g) \Fusion(fg, h) = \al_{\geq}^{(f)} \big( \Fusion(g, h) \big) \Fusion(f, gh)
    \end{equation}
    for all $f, g, h \in G$. Let $\tilde \al_{\geq} = \al_{\geq} \otimes \id$, which is a right restriction for $\tilde \al = \al \otimes \id$ of defect size $2R$. Define unitaries
    \begin{equation*}\label{eq:V_to_representation}
    V(g) := \sum_k \Fusion(g, k) \otimes |k\rangle\langle gk| \in \widetilde \caA_{[j-2R, j + 3R)}, 
    \end{equation*}
    then using Eq. \eqref{eq:trivial anomaly relation} we obtain
    \begin{equation*}
        \tilde \al_{\geq}^{(g)}(V(h))V(g)V(gh)^* = \Fusion(g, h).
    \end{equation*}
    Define a new right restriction $\tilde \beta_{\geq}$ of $\tilde \al$ with components $\tilde \beta_{\geq}^{(g)} = \mathrm{Ad}[V(g)^*] \circ \tilde \al_{\geq}^{(g)}$. Then
    \begin{align*}
    \tilde \beta_{\geq}^{(g)} \circ \tilde \beta_{\geq}^{(h)} &= \mathrm{Ad}[V(g)^*] \circ \tilde \al_{\geq}^{(g)} \circ \mathrm{Ad}[V(h)^*] \circ \tilde \al_{\geq}^{(h)} \\
    &= \mathrm{Ad}[V(g)^* \, \tilde \al_{\geq}^{(g)}(V(h)^*) \, \Fusion(g, h)] \circ \tilde \al_{\geq}^{(gh)} \\
    &= \mathrm{Ad}[V(gh)^*] \circ \tilde \al_{\geq}^{(gh)} = \tilde \beta_{\geq}^{(gh)}.
    \end{align*}
    \ie $g \mapsto \tilde \beta_{\geq}^{(g)}$ is a group homomorphism. Finally noting that $\tilde \beta_{\geq}$ is a right restriction of $\tilde \al$ at $j$ with defect size $5R$ yields the claim.
\end{proof}

%%%%%%%%%%%%%%%%%%%%%%%%%%%%%%%%%%%%%%%%%%%%%%%%%%%
%%%%%%%%%%%%%%%%%%%%%%%%%%%%%%%%%%%%%%%%%%%%%%%%%%%
\subsection{Covariant right restrictions}

Let $\al$ be a symmetry of range $R$ on a spin chain $\caA$, and let $\al_{\geq j}$ be a right restriction of $\al$ at some site $j$. We say $\al_{\geq j}$ is \emph{covariant} if
$$
\al^{(k)} \circ \al_{\geq j}^{(\bar k g k)} \circ \al^{(\bar k)} = \al_{\geq j}^{(g)}
$$
for all $g, k \in G$. The failure of the right restriction to be covariant is captured by local unitaries $\Crossing_g(k)$ which are uniquely defined up to phase by
$$
\al^{(k)} \circ \al_{\geq j}^{(\bar k g k)} \circ \al^{(\bar k)} \circ \big(  \al_{\geq j}^{(g)} \big)^{-1} = \Ad[ \Crossing_g(k) ].
$$
We call these the \emph{crossing operators} associated to the right restriction $\al_{\geq j}$. If $\al_{\geq j}$ has defect size $L$ then $\Crossing_g(k)$ is supported on the interval $[j - (L + 2R), j + (L+2R) \, )$.

By straightforward computation we find
$$
\Ad \left[ \al^{(k)}\big( \Crossing_{\bar k g k}(l) \big) \, \Crossing_{g}(k) \right] = \Ad \left[ \Crossing_{g}(kl) \right]
$$
so there are phases $\lambda_g(k, l) \in U(1)$ such that
$$
\al^{(k)}\big( \Crossing_{\bar k g k}(l) \big) \, \Crossing_{g}(k) = \lambda_g(k, l) \times \Crossing_g(kl)
$$
for all $g, k, l \in G$.

\begin{proposition} \label{prop:obstruction to covariant right restriction}
    The phases $\lambda_g(k, l)$ satisfy the twisted 2-cocycle equations
    $$
        1 = \frac{\lambda_g(k, lm) \lambda_{\bar k g k}(l, m)}{\lambda_g(k, l) \lambda_g(kl, m)}
    $$
    for all $g, k, l, m \in G$. They therefore define a twisted cohomology class $[\lambda] \in H^2(G, U(1)[G])$. (See Appendix \ref{app:cohomology} for the relevant definitions.)

    The class $[\lambda]$ depends only on the symmetry $\al$, \ie it does not depend on the choice of right restriction or the site $j$. So we obtain a well defined map $$\Lambda : \Sym_G \rightarrow H^2(G, U(1)[G])$$ which we call the obstruction to covariant right restrictions.
    
    For symmetries $\al$ and $\beta$ of range $R$ this obstruction satisfies
    \begin{enumerate}
        \item \label{propitem:obstruction decoupled implies trivial} If $\al$ is decoupled then $\Lambda(\al) = [1]$, the identity element of $H^2(G, U(1)[G])$.
        \item \label{propitem:obstruction locally computable}$\Lambda$ is locally computable: If $\al$ and $\beta$ act on the same spin chain $\caA$ and there is an interval $I$ of length $12R$ such that $\al|_{\caA_I} = \beta|_{\caA_I}$ then $\Lambda(\al) = \Lambda(\beta)$.
        \item \label{propitem:obstruction multiplicative} $\Lambda$ is multiplicative under stacking: $\Lambda(\al \otimes \beta) = \Lambda(\al) \cdot \Lambda(\beta)$.
        \item \label{propitem:obstruction constant on stable equivalence classes} $\Lambda$ is constant on stable equivalence classes: $\al \sim \beta \, \implies \, \Lambda(\al) = \Lambda(\beta)$.
    \end{enumerate}
    In particular, the obstruction reduces to a homomorphism of monoids $\Lambda : (\Sym_G / \sim) \rightarrow H^2(G, U(1)[G])$.
\end{proposition}

The proof can be found in Appendix \ref{app:proof of proposition covariant obstruction}.

\begin{remark} \label{rem:remarks on obstruction to covariant right restrictions}
    \begin{enumerate}
        \item It will follow from Theorem \ref{thm:classification} that $(\Sym_G/\sim)$ is a group and so $\Lambda : (\Sym_G / \sim) \rightarrow H^2(G, U(1)[G])$ is in fact a group homomorphism.
    
        \item We could use left restrictions instead of right restrictions to define an obstruction to covariant left restrictions $\Lambda_L(\al)$. Then $\Lambda_L(\al) = \Lambda(\al)^{-1}$.

        \item We will show in Appendix \ref{app:general_covariance} that the obstruction $\Lambda(\al)$ is given by the inverse of the slant product of the anomaly of $\al$ (see Appendix \ref{app:cohomology} for definitions). In the context of SPTs the obstruction $\Lambda$ is therefore intimately related to the types and properties of the anyons supported by the gauged SPT \cite{wang2015topological}, which are believed to be described by the Dijkgraaf-Witten TQFT \cite{dijkgraaf1990topological} corresponding to the anomaly.
    \end{enumerate}
\end{remark}

As the name suggests, if a symmetry $\al$ has vanishing obstruction to covariant right restrictions, then (after stacking with local degrees of freedom) $\al$ indeed admits covariant right restrictions:

\begin{lemma} \label{lem:covariant right restrictions}
    Let $\al$ be a symmetry of range $R$ such that $\Lambda(\al) = [1]$. For any site $j \in \Z$ we can enlarge the on-site algebra $\caA_j$ by stacking with a local degree of freedom $\End(\C^{\abs{G}})$, obtaining an enlarged spin chain $\widetilde \caA$. Then the symmetry $\tilde \al = \al \otimes \rho_{\reg}$, where $\rho_{\reg}$ is the left regular representation of $G$ on the local degree of freedom, admits a covariant right restriction at $j$ with defect size $5R$.
\end{lemma}

\begin{proof}
    Let $\al_{\geq}$ be a right restriction of $\alpha$ at $j \in \Z$ with defect size $2R$. Let $\Crossing_g(k) \in \caA_{[j-4R, j+4R)}$ be the crossing operators associated to this right restriction. Since $\Lambda(\al) = [1]$ we can choose the phases of the $\Crossing_g(k)$ so that
    \begin{equation} \label{eq:trivial obstruction relation}
    \al^{(k)}\big( \Crossing_{\bar k g k}(l) \big) \, \Crossing_{g}(k) = \Crossing_g(kl)
    \end{equation}
    for all $g, k, l \in G$.

    Let $\tilde \al_{\geq} = (\al_{\geq} \otimes \id)$, which is a right restriction of $\tilde \al = \al \otimes \rho_{\reg}$ of defect size $2R$. Define unitaries
    $$
    V_g = \sum_k \,  \Crossing_g(k) \otimes | k \rangle \langle k| \in \widetilde \caA_{[j-4R, j+4R)}.
    $$
    Then
    \begin{align*}
        V_g^* \, \tilde \al^{(k)} \big( V_{\bar k g k } \big) &= \left( \sum_{l_1} \,  \Crossing_g(l_1)^* \otimes | l_1 \rangle \langle  l_1 | \right) \times \left( \sum_{l_2} \,  \tilde \al^{(k)} \big(  \Crossing_{\bar k g k}(l_2) \big) \otimes | k l_2 \rangle \langle k l_2 | \right) \\
        \intertext{putting $l = l_1 = k l_2$ and using Eq. \eqref{eq:trivial obstruction relation} this becomes}
        &= \sum_l \, \Crossing_g(l)^* \, \tilde \al^{(k)} \big( \Crossing_{\bar k g k}(\bar k l ) \big) \otimes | l \rangle \langle l | = \sum_{l} \, \Crossing_g(k)^* \otimes | l \rangle \langle l | = \Crossing_g(k)^*.
    \end{align*}
    Now consider the right restriction $\tilde \beta_{\geq}$ of $\tilde\alpha$ with components $\tilde \beta_{\geq}^{(g)} = \Ad[V_g] \circ \tilde \al_{\geq}^{(g)}$. Then
    \begin{align*}
        \tilde \al^{(k)} \circ \tilde \beta_{\geq}^{(\bar k g k)} \circ \tilde \al^{(\bar k)} \circ \big( \tilde \beta_{\geq}^{(g)} \big)^{-1} &= \tilde \al^{(k)} \circ \Ad[V_{\bar k g k}] \circ \tilde \al_{\geq}^{(\bar k g k)} \circ \tilde \al^{(\bar k)} \circ \big( \tilde \al_{\geq}^{(g)} \big)^{-1} \circ \Ad[ V_g^* ] \\
        &= \Ad \left[ \tilde \al^{(k)}\big( V_{\bar k g k} \big) \Crossing_g(k) \, V_g^* \right] = \id,
    \end{align*}
    so $\tilde \beta_{\geq}$ is covariant. Noting that $\tilde \beta_{\geq}$ is a right restriction of $\tilde \al$ at $j$ with defect size $5R$ finishes the proof.
\end{proof}

We now show that symmetries with trivial anomaly have no obstruction to covariant right restrictions.
\begin{lemma} \label{lem:trivial anomaly implies trivial obstruction}
    If $\al$ has trivial anomaly then $\Lambda(\al) = [1]$.
\end{lemma}

\begin{proof}
    Suppose $\al$ has range $R$. We can stack the on-site algebra $\caA_j$ by $\End(\C^{\abs{G}})$, obtaining an extended spin chain $\widetilde \caA$. By Lemma \ref{lem:group hom restrictions} there is a right restriction $\tilde \al_{\geq}$ of $\tilde \al = \al \otimes \id$ at $j$ with defect size $5R$  and such that $g \mapsto \tilde \al_{\geq}^{(g)}$ is a group homomorphism. We can therefore regard $\tilde \al_{\geq}$ as a symmetry.

    Since $\tilde \al_{\geq}$ and $\al$ agree everywhere to the right of the site $j+5R$ it follows from local computability of the obstruction to covariant right restrictions (item \ref{propitem:obstruction locally computable} of Proposition \ref{prop:obstruction to covariant right restriction}) that $\Lambda(\al) = \Lambda(\tilde \al_{\geq})$. But $\tilde \al_{\geq}$ agrees with $\id$ everywhere to the left of the site $j-5R$, so $\Lambda(\tilde \al_{\geq}) = \Lambda(\id) = [1]$ by local computability. We conclude that $\Lambda(\al) = [1]$.
\end{proof}

%%%%%%%%%%%%%%%%%%%%%%%%%%%%%%%%%%%%%%%%%%%%%%%%%%%
\subsection{Covariant right restrictions that are group homomorphisms}

\begin{lemma} \label{lem:trivial anaomly allows covariant left restrictions that are group homs}
    Let $\al$ be a symmetry of range $R$ with trivial anomaly on a spin chain $\caA$. For any site $j \in \Z$ we can enlarge the on-site algebra $\caA_j$ by stacking with three local degrees of freedom $\End(\C^{\abs{G}}) \otimes \End( \C^{\abs{G}} ) \otimes \End(\C^{\abs{G}})$, obtaining an enlarged spin chain $\caA'''$. There is an extension $ \al''' = \al \otimes \rho_{\reg}^{\otimes 2} \otimes \rho_{\ad}$ of $\al$ to the enlarged spin chain and a covariant right restriction $\beta_{\geq}$ of $ \al'''$ at $j$ with defect size $7R$ such that $g \mapsto \beta_{\geq}^{(g)}$ is a group homomorphism. Here $\rho_{\ad}^{(k)} = \Ad[J^{(k)}]$ with $J^{(k)} | h \rangle = | k h \bar k \rangle$.
\end{lemma}

\begin{proof}
    Fix $j \in \Z$. We stack the on-site algebra $\caA_j$ with a local degree of freedom $\End(\C^{\abs{G}})$, obtaining an enlarged spin chain $\caA'$. Let $a = \al \otimes \rho_{\reg}$ be the enlarged symmetry on $\caA'$. Then, Lemmas \ref{lem:covariant right restrictions} and \ref{lem:trivial anomaly implies trivial obstruction} imply that there is a covariant right restriction $r$ of $a$ at $j$ of defect size $5R$.
    
    Let $j' = j - 12R$ be far to the left of $j$. By stacking with an ancilla at $j'$ (which is purely auxiliary, and does not appear in the final result), Lemma \ref{lem:group hom restrictions} gives a right restriction $r'$ of $a$ at $j'$ with defect size $5R$, which is a group homomorphism. We extend $a$ and $r$ to act as identity on this ancilla. We have $r'^{(k)} = a^{(k)}$ on $\caA_{\geq j' + 5R}$ for all $k \in G$.

    Since $r'$ agrees with $a$ to the right, and $r$ is a covariant right restriction of $a$, we have
    \begin{equation} \label{eq:pseudo-covariance}
        r'^{(k)} \, r^{(\bar k g k)} \, r'^{(\bar k)} = r^{(g)}.
    \end{equation}
    for all $g, k \in G$.

    Since $r$ and $r'$ are both right restrictions of $a$, there are local unitaries $T_g$ such that
    $$ r^{(g)} = \Ad[T_g] \circ r'^{(g)}. $$
    Using the fact that $r'$ is a group homomorphism, the operators
    $$ \Phi(g, h) := T_g r'^{(g)}( T_h ) T_{gh}^* $$
    can be shown to be fusion operators for $r$. Moreover, by straightforward computation, they satisfy
    \begin{equation} \label{eq:strict associativity}
        \Phi(f, g) \Phi(fg, h) = r^{(f)} \big( \Phi(g, h) \big) \, \Phi(f, gh).
    \end{equation}

    Inserting $r^{(g)} = \Ad[T_g] \circ r'^{(g)}$ into Eq. \eqref{eq:pseudo-covariance} yields $\Ad[ r'^{(k)}( T_{\bar k g k} ) ]=  \Ad[T_g]$, so there are phases $c_g(k)$ such that
    \begin{equation} \label{eq:transformation of T}
        r'^{(k)}( T_{\bar k g k} ) = c_g(k) T_g
        \end{equation}
    for all $g, k \in G$. Again using the fact that $r'$ is a group homomorphism then implies that these phases satisfy a twisted cocycle relation (cf. Eq. \eqref{eq.twisted-1cocycle}):
    \begin{equation}\label{eq:cocycle relation for c}
        1 = \frac{ c_g(k) c_{\bar k g k}(l) }{ c_g(kl) }.
    \end{equation}
    Moreover, applying $r'^{(k)}$ to the fusion operator yields
    $$ r'^{(k)}\big( \Phi( \bar k g k , \bar k h k ) \big)  = \frac{ c_g(k) c_h(k) }{ c_{gh}(k) }  \Phi(g, h).$$

    Since the $\Phi$ are fusion operators for $r$ they are supported on $[ j-5R, j + 6R )$. The right restriction $r'$ therefore agrees with $a$ on these fusion operators and we obtain
    \begin{equation} \label{eq:trivial covariance obstruction}
        a^{(k)}\big( \Phi( \bar k g k , \bar k h k ) \big)  = \frac{ c_g(k) c_h(k) }{ c_{gh}(k) }  \Phi(g, h).
    \end{equation}

    We can now restrict $a$ and its covariant right restriction $r$ back to the spin chain $\caA'$ without the ancilla, retaining the relations Eq \eqref{eq:strict associativity} and Eq. \eqref{eq:trivial covariance obstruction}. Indeed, the fusion operators $\Phi(f, g)$ are fusion operators for $r$, and are supported on $[j - 5R, j + 6R)$. In particular, they act trivially on the ancilla. We have now a covariant right restriction $r$ whose fusion operators satisfy Eq. \eqref{eq:trivial covariance obstruction} with phases $c_{g}(k)$ satisfying the cocycle relation Eq. \eqref{eq:cocycle relation for c}. We now stack with another ancilla to remove this cocycle.

    Add a local degree of freedom $\End(\C^{\abs{G}})$ at site $j$, we obtain an enlarged spin chain $\caA''$. We extend the symmetry to $\al'' = a \otimes \rho_{\reg} = \al \otimes \rho_{\reg} \otimes \rho_{\reg}$, and take a right restriction $\al''_{\geq}$ with components $(\al_{\geq}'')^{(g)} = r^{(g)} \otimes \Ad[ U(g) ]$ where
    $$ U(g) = \sum_{l \in G} c_g(l) | l \rangle \langle l |.$$
    Using the twisted 1-cocycle relation \eqref{eq:cocycle relation for c} satisfied by $c$, one easily checks that $\rho_{\reg}^{(k)} (U(\bar kg k)) = c_g(k)^{-1} U(g)$. It follows that $\al''_{\geq}$ is a covariant right restriction of $\al''$, with fusion operators
    $$ \Phi''(g, h) =\Phi(g, h) \otimes U(g) U(h)  U(gh)^*. $$
    Together with Eq. \eqref{eq:trivial covariance obstruction} this yields
    \begin{equation} \label{eq:covariance''}
    (\al'')^{(k)}\big( \Phi''(\bar k g k, \bar k h k) \big) = \Phi''(g, h).
    \end{equation}
    Moreover, Eq. \eqref{eq:strict associativity} implies
    $$ \Phi''(f, g) \Phi''(fg, h) = (\al''_{\geq})^{(f)} \big( \Phi''(g, h) \big) \, \Phi''(f, gh). $$
    We now add another local degree of freedom $\End{\C^{\abs{G}}}$ at site $j$ and, as in the proof of Lemma \ref{lem:group hom restrictions}, we define unitaries $V(g) = \sum_{l} \Fusion''(g, l) \otimes |l \rangle \langle g l |$ supported on $[j-5R, j+6R)$ so that $\beta_{\geq}^{(g)} = \Ad[ V(g)^* ] \circ (\al''_{\geq} \otimes \id)$ defines a right restriction of $\al''' = \al'' \otimes \rho_{\ad}$, where $g \mapsto \beta_{\geq}^{(g)}$ is a group homomorphism. 

    Moreover, using Eq. \eqref{eq:covariance''} we get
    $$
    (\al''')^{(k)} \big( V(\bar k g k) \big) = \sum_{l \in G}  \Fusion''(g, l) \otimes | l \rangle \langle gl | = V(g).
    $$
    This implies that $\beta_{\geq}^{(g)} = \Ad[ V(g)^* ] \circ (\al''_{\geq} \otimes \id)$ is a covariant right restriction of $\al'''$ at $j$, which is also a group homomorphism. Since the $V(g)$ are supported on $[j - 5R, j+6R)$, it also follows that the defect size can be taken as $7R$.
\end{proof}

%%%%%%%%%%%%%%%%%%%%%%%%%%%%%%%%%%%%%%%%%%%%%%%%%%%
%%%%%%%%%%%%%%%%%%%%%%%%%%%%%%%%%%%%%%%%%%%%%%%%%%%
%%%%%%%%%%%%%%%%%%%%%%%%%%%%%%%%%%%%%%%%%%%%%%%%%%%
\section{Proof of the main Theorem \ref{thm:classification}} \label{sec:proof of main thm}

In this section we prove injectivity of the map $\Omega : (\Sym_G/\sim) \,  \rightarrow H^3(G, U(1))$ (Corollary \ref{cor:same anomaly implies stably equivalent}), and prove the main classification Theorem \ref{thm:classification}. 

\subsection{Injectivity of $\Omega$}

Firstly, we prove that a symmetry with trivial anomaly is stably equivalent to a symmetry that can be written as a product of mutually commuting local representations:

\begin{lemma} \label{lem:trivial anomaly implies commuting product expansion}
Let $\alpha$ be a symmetry on a spin chain $\caA$, with trivial anomaly. Then $\alpha \sim \alpha''$, where $\alpha''$ is a symmetry that can be written as a formal product
\begin{equation}
    (\al'')^{(g)} = \prod_{j \in \Z} (\al'')_{j}^{(g)},  \qquad (\al'')^{(g)}_j = \Ad(\tilde U_j(g)),
\end{equation}
where
 \begin{enumerate}
       \item  $g \mapsto  (\al_{j}'')^{(g)}$ is a group homomorphism,
       \item $(\al'')^{(g)}_j$ is supported on $[jL - 7R, (j+1)L + 8R]$ for all $g \in G$,
       \item    $[\tilde U_i(g),\tilde U_j(h)]=0$, whenever $i\neq j$, for all $g,h \in G$. 
   \end{enumerate}
\end{lemma}
\begin{proof}
    Let $R$ be the range of $\al$ and take $L = 16R$. We stack $\caA$ with local degrees of freedom $\End(\mathbb{C}^{|G|})^{\otimes 3}$ at every site $jL$ for $j \in\Z$, obtaining a spin chain $\caA'$. Then, by Lemma \ref{lem:trivial anaomly allows covariant left restrictions that are group homs}, there is a symmetry $ \al' \sim \al$ of range $R$ on the enlarged spin chain $ \caA'$ that admits covariant right restrictions $ \al_{\geq jL}'$ of defect size $7R$ at sites $jL$ for all $j\in \Z$, and such that $g \mapsto  \al_{\geq jL}'^{(g)}$ are group homomorphisms.

    Covariance implies that whenever $j \geq i + 1$ we have
    \begin{equation} \label{eq:consequence of covariance}
     \al_{\geq iL}'^{(k)} \circ \al_{\geq jL}'^{(\bar k g k)} \circ \al_{\geq iL}'^{(\bar k)} =  \al'^{(k)} \circ  \al_{\geq jL}'^{(\bar k g k)} \circ \al'^{(\bar k)} =  \al_{\geq jL}'^{(g)}.
    \end{equation}
    For each $j \in \Z$, define $ \al_{j}'^{(g)} :=  \al_{\geq jL}'^{(g)} \circ ( \al_{\geq (j+1)L}'^{(g)}) ^{-1}$. Then 
    \begin{enumerate}
        \item Eq. \eqref{eq:consequence of covariance} implies that  $g \mapsto \al'^{(g)}_j$ are group homomorphisms,
        \item by definition, each $\alpha_j'$ is supported on $[jL-7R,(j+1)L+8R]$,
        \item covariance also implies that $ [\al_{i}'^{(g)}, \al_{j}'^{(h)}]=0$ whenever $i\neq j$.
    \end{enumerate}
    Pick unitaries $U_j(g) \in  \caA'$ such that $ \al_{j}'^{(g)} = \Ad[U_j(g)]$. Since $g \mapsto  \al_{j}'^{(g)}$ is a group homomorphism, the unitaries $U_j(g)$ form a projective representation of $G$. The unitaries $U_j(g)$ commute with unitaries $U_{j'}(h)$ whenever $\abs{j' - j} \geq 2 $ 
    because they have disjoint supports (see figure \ref{fig:almost commuting unitaries for trivial anomaly}).

%-------------------------------------------------------------

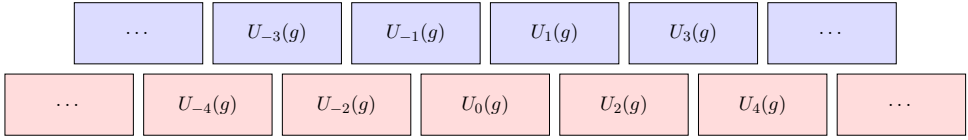
\begin{figure} [h]
\begin{center}
\resizebox{.8\linewidth}{!}{
\begin{tikzpicture}
    % Horizontal spacing and sizes
    \def\xdist{2.7}
    \def\boxwidth{2.5}
    \def\boxheight{1.2}
    \def\ydist{1.4}

    % Define soft colors
    \definecolor{softred}{RGB}{255,220,220}
    \definecolor{softblue}{RGB}{220,220,255}

    % Main even-indexed bottom layer: j from -4 to 4
    \foreach \j [evaluate=\j as \label using int(\j)] in {-4,-2,0,2,4} {
        \pgfmathsetmacro{\x}{(\j + 6)/2 * \xdist}
        \node[draw=black, fill=softred, text=black, minimum width=\boxwidth cm, minimum height=\boxheight cm] 
            at (\x, 0) {$U_{\label}(g)$};
    }

    % Ellipsis boxes on the left and right of bottom layer
    \node[draw=black, fill=softred, text=black, minimum width=\boxwidth cm, minimum height=\boxheight cm] 
        at (0, 0) {$\dots$};
    \node[draw=black, fill=softred, text=black, minimum width=\boxwidth cm, minimum height=\boxheight cm] 
        at (6*\xdist, 0) {$\dots$};

    % Main odd-indexed top layer: j from -3 to 3
    \foreach \j [evaluate=\j as \label using int(\j)] in {-3,-1,1,3} {
        \pgfmathsetmacro{\x}{(\j + 6)/2 * \xdist}
        \node[draw=black, fill=softblue, text=black, minimum width=\boxwidth cm, minimum height=\boxheight cm] 
            at (\x, \ydist) {$U_{\label}(g)$};
    }

    % Ellipsis boxes on the left and right of top layer
    \node[draw=black, fill=softblue, text=black, minimum width=\boxwidth cm, minimum height=\boxheight cm] 
        at (\xdist/2, \ydist) {$\dots$};
    \node[draw=black, fill=softblue, text=black, minimum width=\boxwidth cm, minimum height=\boxheight cm] 
        at (5.5*\xdist, \ydist) {$\dots$};

\end{tikzpicture}}
\end{center}
    \caption{The symmetry $\alpha'$ can be seen as a conjugation by a FDQC $\prod\limits_{j\in \mathbb Z} U_j(g)$.}
    \label{fig:almost commuting unitaries for trivial anomaly}
\end{figure}

%-------------------------------------------------------------

However, since $\al_{j-1}'^{(g)}$ commutes with $\al_{j}'^{(h)}$ we find that there are phases $\chi_j : G^2 \rightarrow U(1)$ such that
        \begin{equation}
            \label{eq: unitaries up to phase}
               U_{j-1}(g) U_{j}(h) = \chi_j(g, h) \, U_{j}(h) U_{j-1}(g).
        \end{equation}
    
    In order to obtain the statements of the lemma, our goal is now to upgrade these unitaries $U_j$ to unitaries $\widetilde  U_j$ with the same spatial support, and such that  $g\mapsto \widetilde  U_j(g)$ are still projective representations, but also such that 
    \begin{equation} \label{eq:unitaries commute}
        [\widetilde  U_j(g),\widetilde  U_i(h)]=0, \qquad i\neq j,
    \end{equation}
    instead of \eqref{eq: unitaries up to phase}. 
    
    In order to achieve this, we need to stack again. 
    At each site $jL$ with $j \in \Z$, we add a local degree of freedom $\caC_j$ which is a copy of  $\caA'_{[L(j-1) - 7R, L(j+1) + 8R]}$, obtaining the enlarged spin chain $\caA''$. 
    We recall that the projective representations $U_{j-1}$ and $U_j$ map in $\caA'_{[L(j-1) - 7R, L(j+1) + 8R]}$. Let us now consider the conjugate representations $\overline U_{j-1}$ and $\overline U_{j}$ mapping in the copy $\caC_j$.  The conjugate representations will always be considered inside $\caC_j$ whereas the original representations act on the spin chain $\caA'$. 
    We now define the symmetry $\rho_{Lj}''$ on the added algebras in $\caC_j$ given by
    $$
    (\rho''_{jL})^{(g)} = \Ad[ \overline U_{j-1}(g) \overline U_{j}(g) ]
    $$
    To check that this is indeed a representation, we recall that $U_{j-1},U_j$, and hence also $\overline U_{j-1},\overline U_j$, commute up to a phase.  In fact, we have 
        \begin{equation}
            \label{eq: conj unitaries up to phase}
    \overline U_{j-1}(g) \overline U_j(h) = \bar \chi_j(g, h) \, \overline U_j(h) \overline U_{j-1}(g). 
    \end{equation}
     
    Writing $\rho'' = \bigotimes_{j \in \Z} \, \rho''_{jL}$ we thus obtain a new symmetry $\al'' = \al' \otimes \rho'' \sim \al'$ acting on the new spin chain $\caA''$. 
    Finally, we define 
    $$\widetilde U_j(g) =  \overline U_j(g) \otimes U_j(g) \otimes \overline U_{j}(g)  \in \caC_{j} \otimes \caA' \otimes  \caC_{j+1} 
    $$
    (see figure \ref{fig:trivial anomaly - commuting unitaries}).

%-------------------------------------------------------------
\begin{figure} [h]
    \centering

\resizebox{.8\linewidth}{!}{

\begin{tikzpicture}
    % Parameters
    \def\splitgap{0.2}  % Vertical gap between split squares
    \def\xdist{2.7}
    \def\boxwidth{2.5}
    \def\boxheight{1.2}
    \def\ydist{1.4}
    \def\greenheight{3.1}
    \def\greenwidth{0.35} % Narrow green boxes like in the image
    \def\greendist{2.55}  % Higher vertical spacing to avoid overlap

    % Define soft colors
    \definecolor{softred}{RGB}{255,220,220}
    \definecolor{softblue}{RGB}{220,220,255}
    \definecolor{softgreen}{RGB}{220,255,220}

    % vertical lines: sites jL (feel free to tidy up)

    \foreach \i [evaluate=\i as \j using int(\i - 7)] in {2,3,4,5,6,7,8,9,10,11,12,13} {
    \pgfmathsetmacro{\x}{\i/2 * \xdist - 3*\xdist/4}
    \pgfmathsetmacro{\halfheight}{\greenheight/2}
    \draw (\x,\greendist + \ydist + \halfheight/2 + \splitgap/2) -- (\x,-0.8);
    \node[label] at (\x,-1.05) {$\j L$};
    }

    % Layer 1: bottom, even indices
    \foreach \j [evaluate=\j as \label using int(\j)] in {-4,-2,0,2,4} {
        \pgfmathsetmacro{\x}{(\j + 6)/2 * \xdist}
        \node[draw=black, fill=softred, text=black, minimum width=\boxwidth cm, minimum height=\boxheight cm] 
            at (\x, 0) {$U_{\label}(g)$};
    }

    % Layer 2: middle, odd indices
    \foreach \j [evaluate=\j as \label using int(\j)] in {-5,-3,-1,1,3,5} {
        \pgfmathsetmacro{\x}{(\j + 6)/2 * \xdist}
        \node[draw=black, fill=softblue, text=black, minimum width=\boxwidth cm, minimum height=\boxheight cm] 
            at (\x, \ydist) {$U_{\label}(g)$};
    }

% Layer 3: green boxes (split), with alternating softred/softblue colors and shifted top labels
\foreach \i [evaluate=\i as \j using int(\i - 7)] in {2,3,4,5,6,7,8,9,10,11,12,13} {
    \pgfmathsetmacro{\x}{\i/2 * \xdist - 3*\xdist/4}
    \pgfmathsetmacro{\halfheight}{\greenheight/2}

    % Color logic: top = red if even, blue if odd; bottom = opposite
    \ifodd\i
        \def\topcolor{softblue}
        \def\bottomcolor{softred}
    \else
        \def\topcolor{softred}
        \def\bottomcolor{softblue}
    \fi

    % Bottom square with label U_j(g)
    \node[draw=black, fill=\bottomcolor, minimum width=\greenwidth cm, minimum height=\halfheight cm, align=center]
        at (\x, \greendist + \ydist - \halfheight/2 - \splitgap/2)
        {\rotatebox{90}{$\overline U_{\j}(g)$}};
    
    % Top square with label U_{j-1}(g)
    \node[draw=black, fill=\topcolor, minimum width=\greenwidth cm, minimum height=\halfheight cm, align=center]
        at (\x, \greendist + \ydist + \halfheight/2 + \splitgap/2)
        {\rotatebox{90}{$\overline U_{\number\numexpr\j-1\relax}(g)$}};
}

% Bottom layer ellipsis
\node[draw=black, fill=softred, minimum width=\boxwidth cm, minimum height=\boxheight cm]
    at (0, 0) {$\dots$};
\node[draw=black, fill=softred, minimum width=\boxwidth cm, minimum height=\boxheight cm]
    at (6*\xdist, 0) {$\dots$};

% Middle layer ellipsis
\node[draw=black, fill=softblue, minimum width=\boxwidth cm, minimum height=\boxheight cm]
    at (\xdist/2, \ydist) {$\dots$};
\node[draw=black, fill=softblue, minimum width=\boxwidth cm, minimum height=\boxheight cm]
    at (5.5*\xdist, \ydist) {$\dots$};
\end{tikzpicture}}
    \caption{The unitaries $\tilde U_j(g)$ are defined as tensor products of $U_j(g)$ with projective representations $\overline U_j(g)$ defined on appropriate degrees of freedom $\mathcal C_j,\mathcal C_{j+1}$. }
    \label{fig:trivial anomaly - commuting unitaries}
\end{figure}
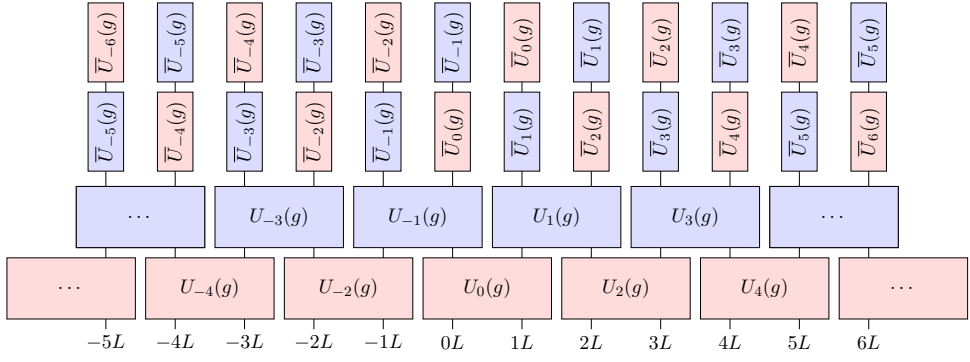
%-------------------------------------------------------------
    
    We note that $\widetilde U_j$  is a projective representation since it is a tensor product of projective representations. 
    Writing $\al_{j}''^{(g)} := \Ad[ \widetilde U_j(g) ]$, we can verify that all properties of $\al_j'$ listed below \eqref{eq:consequence of covariance} hold for $\al_j''$ as well, in particular, 
    $\al'' = \prod_{j \in \Z} \al_{j}''$.  
    However, as announced, we now have the stronger commutation property \eqref{eq:unitaries commute}.
    This is checked by an explicit computation using \eqref{eq: unitaries up to phase}, \eqref{eq: conj unitaries up to phase} and $ \chi_j(g, h)\bar \chi_j(g, h)=1$.  

    \end{proof}

\begin{proposition} \label{prop:trivial anomaly implies stably equivalent to decoupled symmetry}
    Let $\al$ be a symmetry with trivial anomaly. Then $\al$ is stably equivalent to a decoupled symmetry.
\end{proposition}
\begin{proof}
By Lemma \ref{lem:trivial anomaly implies commuting product expansion}, the symmetry $\al$ is stably equivalent to a symmetry $\alpha''$ that splits into local components which commute with each other. We proceed in constructing an explicit equivalence between $\alpha''$ and a block partitioned QCA. Let $R$ denote the range of $\alpha$, and let $L=16R$, exactly as in the proof of the previous Lemma. We use freely the notation from Lemma \ref{lem:trivial anomaly implies commuting product expansion}. 
      
    The decoupling of $\al''$ requires more stacking.
    For each $j \in \Z$ let $m(j) = jL + 8R$ be the site sitting half way between $jL$ and $(j+1)L$. Stack with a spin chain $\caB$ with on-site algebras $\caB_{m(j)} = \End(\C^{\abs{G}})$ for all $j$ and $\caB_{j'} \simeq \C$ at all other sites $j'$ to obtain a new spin chain $ \caA'''$. Let $\rho'''_{\reg}$ be the decoupled symmetry on $\caB$ which acts with the left regular representation $\rho_{\reg}$ on each $\caB_{m(j)}$. Let $ \al''' =  \al'' \otimes \rho'''_{\reg} \sim  \al'' \sim \al$. We now explicitly construct a FDQC $\gamma$ such that $\gamma^{-1} \circ \al''' \circ \gamma$ is block partitioned.
    Define unitaries
    $$
     V_j := \sum_{h \in G} \, \widetilde U_j(h) \otimes | h \rangle \langle h |
    $$
    where the second tensor factor corresponds to $\caB_{m(j)}$. Let $\gamma_j = \Ad[V_j]$, then $\gamma_j$ has the same support as $ \al''_{j}$. Since $\widetilde U_j$ is a projective representation we have
    \begin{align*}
        W_j(g) := V_j^* \, \big( \widetilde U_j(g) \otimes L^{m(j)}_{g} \big) \, V_j = \sum_{h \in G} \widetilde U_j(gh)^* \widetilde U_j(g) \widetilde U_j(h) \otimes |gh \rangle \langle h| = \sum_{h \in G} \, \I \, \otimes \, c_j(g, h) | gh \rangle \langle h|,
    \end{align*}
    for some 2-cocycle $c_j : G^2 \rightarrow U(1)$. Here $g \mapsto L_g^{m(j)}$ is the left regular representation on the ancilla at site $m(j)$, we write $(\rho_{\reg})_{m(j)}^{(g)}$ for its adjoint action. The unitary $W_j(g)$ is seen to act non-trivially only on $\caB_{m(j)}$ and so
    $$
    \gamma_j^{-1} \circ \big(  \al_{j}''^{(g)} \, \otimes \,  \rho_{\reg, m(j)}^{(g)} \big) \circ  \gamma_j = \Ad[W_j(g)]
    $$
    is a $G$-action which acts non-trivially only on $\caB_{m(j)}$.   
    Moreover, from Eq. \eqref{eq:unitaries commute} we find that each $\gamma_j$ commutes with every component of $\al_{i}''$ whenever $i \neq j$. Defining the FDQC $\gamma = \prod_{j \in \Z} \gamma_j$ and recalling that $ \al'' = \prod_{j \in \Z}  \al''_{j}$ we therefore find that
    $$
    \beta^{(g)} := \gamma^{-1} \circ  \al'''^{(g)} \circ \gamma = \gamma^{-1} \circ ( \al''^{(g)} \otimes \rho_{\reg}'''^{(g)}) \circ \gamma = \prod_{j \in \Z} \Ad[W_j(g)].
    $$
    We conclude that $\al \sim \beta = \bigotimes_{j \in \Z} \, \Ad[W_j]$ is a decoupled symmetry.
\end{proof}

\begin{corollary} \label{cor:same anomaly implies stably equivalent}
    The map $\Omega: (\Sym_G / \sim) \rightarrow H^3(G,U(1))$ is injective. In other words, let $\al$ and $\beta$ be symmetries on spin chains $\caA$ and $\caB$, respectively. If $\al$ and $\beta$ have the same anomaly, then they are stably equivalent.
\end{corollary}

\begin{proof}
    Let $\gamma$ be a symmetry with anomaly $\ano(\gamma) = \ano(\al)^{-1}$, defined on a spin chain $\caC$. (One could take the symmetry constructed in Section \ref{sec:examples}).

    Using multiplicativity of the anomaly under stacking and Proposition \ref{prop:trivial anomaly implies stably equivalent to decoupled symmetry} we find
    $$
    \al \sim \al \otimes (\gamma \otimes \beta) = (\al \otimes \gamma) \otimes \beta \sim \beta
    $$
    as required.
\end{proof}

\subsection{Classification of locality preserving symmetries}

We are now ready to prove our main Theorem \ref{thm:classification}.\\

\begin{proofof}[Theorem \ref{thm:classification}]
    To show that the monoid $(\Sym_G / \sim)$ is in fact a group it suffices to show that each class has an inverse. Let $\al$ be a symmetry. Then by Section \ref{sec:examples} there exists a symmetry $\beta$ such that $\Omega(\beta) = \Omega(\al)^{-1}$. Since the anomaly is multiplicative under stacking we have that $\al \otimes \beta$ has trivial anomaly. By Proposition \ref{prop:trivial anomaly implies stably equivalent to decoupled symmetry} the stacked symmetry $\al \otimes \beta$ is stably equivalent to a decoupled symmetry. This shows that the stable equivalence class of $\beta$ is the inverse of the class of $\al$ in $(\Sym_G / \sim)$.

    The statement that two symmetries are stably equivalent if and only if their anomalies are equal follows from item 4 of  Proposition \ref{prop:anomaly} and Corollary $\ref{cor:same anomaly implies stably equivalent}$. The fact that $\Omega$ lifts to a group homomorphism from $(\Sym_G/\sim)$ to $H^3(G, U(1))$ also follows from Proposition \ref{prop:anomaly}. That this homomorphism is actually an isomorphism follows from the examples of Section \ref{sec:examples}, which provide for each $[\omega] \in H^3(G, U(1))$ a symmetry $\al$ with $\ano(\al) = [\omega]$.
\end{proofof}

%% APPENDICES
\appendix
\addappheadtotoc
\addtocontents{toc}{\protect\setcounter{tocdepth}{0}} % Remove appendix sections from ToC

%%%%%%%%%%%%%%%%%%%%%%%%%%%%%%%%%%%%%%%%%%%%%%%%%%%
%%%%%%%%%%%%%%%%%%%%%%%%%%%%%%%%%%%%%%%%%%%%%%%%%%%
%%%%%%%%%%%%%%%%%%%%%%%%%%%%%%%%%%%%%%%%%%%%%%%%%%%
\section{(Twisted) group cohomology} \label{app:cohomology}

\subsection{Group cohomology}

We give the necessary definitions of (twisted) group cohomology and the slant product. For an in depth treatment, see for example the monograph \cite{brown2012cohomology}. Recall that for a group $G$, a (left) \textit{$G$-module} is an abelian group $M$ equipped with a left $G$-action $\_ \cdot \_\ : G \times M \longrightarrow M$ satisfying
   $$g \cdot (xy) = (g\cdot x)(g\cdot y)$$
for all $x,y \in M$, $g \in G$, where we use multiplicative notation for the module operation. Let $G$ be a group and $M$ a $G$-module. Consider the cochain complex 
\begin{align*}
    \sdots\overset{\delta}{\longrightarrow} C^n \overset{\delta}{\longrightarrow} C^{n+1} \overset{\delta}{\longrightarrow} \sdots
\end{align*}
where $C^n:= C^n(G,M)$ is the collection of all functions $G^n \to M$ and $\delta$ is the differential map which sends a function $C^n \ni \theta: G^n \to M$ to $\delta \theta \in C^{n+1}$, defined by 
\begin{equation} \label{eq.differential-map}
    (\delta \theta)(g_1,\sdots,g_{n+1}) = g_1 \cdot \theta(g_2,\sdots, g_{n+1})\left(\prod\limits_{j=1}^n \theta(g_1, \sdots, g_j g_{j+1}, \sdots, g_{n+1})^{(-1)^j} \right)\theta(g_1,\sdots, g_n)^{(-1)^{n+1}}. 
\end{equation} 
The differential map $\delta$ is also called the coboundary map, and it satisfies $(\delta \circ \delta)(\theta) = 1$. The \emph{$n$-cocycles} $Z^n$ are the functions in $C^n$ in the kernel of the coboundary map $\delta$.
The \emph{$n$-coboundaries} $B^n := B^n(G,M)$ are the functions of $C^n$ in the image of $\delta : C^{n-1}\to C^n$. 
The \emph{$n$-th cohomology group} $H^n(G,M)$ is defined by $H^n(G,M) = Z^n / B^n$ together with the multiplication $[\theta ][\phi] = [\theta\phi]$ where $(\theta\phi)(g_1,\sdots,g_n) = \theta(g_1,\sdots,g_n)\phi(g_1,\sdots,g_n)$ for two representatives $\theta, \phi \in Z^n$. The quotient is well defined because of the property $(\delta \circ \delta)(\theta)=1$. 

In the following, and throughout the rest of this paper, the words \emph{cocycle}, \emph{coboundary} and \emph{cohomology} will always refer to the case $M=U(1)$ with the trivial action of $G$ on $U(1)$, namely, the action $g \cdot \theta = \theta$ for all $g \in G$.

\subsubsection{Third group cohomology}

Since the anomaly of a locality preserving symmetry takes values in degree three group cohomology, we explicitly state both the cocycle and coboundary conditions: A map $\omega: G^3\to U(1)$ is a (normal) 3-cocycle if 
\begin{equation}\label{eq:3cocycle}
    (\delta \omega) (g_1,g_2,g_3,g_4) = \frac{\omega(g_2,g_3,g_4)\omega(g_1,g_2g_3,g_4)\omega
    (g_1,g_2,g_3)}{\omega(g_1g_2,g_3, g_4), \omega(g_1,g_2,g_3g_4)} = 1.
\end{equation}
Two 3-cocycles are equivalent (represent the same cohomology class) if they are equal up to multiplication by a 3-coboundary
\begin{equation}\label{eq:3coboundary}
  (\delta\xi)(g_1,g_2,g_3) = \frac{\xi(g_2,g_3)\xi(g_1,g_2g_3)}{\xi(g_1g_2,g_3)\xi(g_1,g_2)}
\end{equation}
for some $\xi:G^2 \to U(1)$. 

\subsection{Twisted group cohomology}

The $G$-module $U(1)[G]$ consists of maps $\lambda: G \to U(1)$, assigning $g\mapsto\lambda_g$, with a left $G$-action
$$(k \cdot \lambda)_g = \lambda_{\bar k g k}, \qquad k,g \in G.$$ This is a left action as $(k\cdot l\cdot \lambda)_g = (l\cdot \lambda)_{\bar k g k } = \lambda_{\overline{kl}g kl } = (kl\cdot\lambda)_g$.
Abelian multiplication is given pointwise by $(\lambda\lambda')_g := \lambda_g\lambda'_g$, where $\lambda_g\lambda'_g$ is $U(1)$ multiplication. The module $U(1)[G]$ is a $G$-graded crossed module
$\bigoplus_{g \in G} U(1)_g,$
where  each $U(1)_g$ is a copy of the unitary group $U(1)$ and $G$ acts as conjugation on the label, \ie  $k \cdot U(1)_g \subseteq U(1)_{kg\bar k}$ \cite{etingof2015tensor, propitius1995topological}.

\begin{example} [First degree twisted cohomology]

By applying the definition of the differential \eqref{eq.differential-map}, we check that a map $c: G \to U(1)[G]$ is a twisted 1-cocycle if
\begin{equation}\label{eq.twisted-1cocycle}
    (\delta c)_{h} (g_1,g_2) = \frac{c_{\bar g_1 h g_1}(g_2) c_h(g_1) }{c_h(g_1g_2)} = 1,\qquad h,g_1,g_2 \in G, 
\end{equation}
and a twisted 1-coboundary is a map
\begin{equation}
    (\delta \eta)_h(g) = \frac{\eta_{\bar ghg}}{\eta_h}, \qquad g,h \in G,
\end{equation}
for some $\eta \in U(1)[G]$. 
\end{example}

\begin{example} [Second degree twisted cohomology]

A map $\lambda: G^2\to U(1)[G]$ is a twisted 2-cocycle if
\begin{equation}\label{eq:twisted2cocycle}
    \frac{\lambda_{\bar g_1 h g_1}(g_2,g_3)\lambda_h(g_1,g_2g_3)}{\lambda_h(g_1g_2,g_3)\lambda_h(g_1,g_2)} = 1, \qquad \text{for all } h, g_1,g_2,g_3 \in G.
\end{equation}
Two twisted 2-cocycles are equivalent (represent the same twisted cohomology class) if they are equal up to multiplication by a twisted 2-coboundary, i.e., a map 
\begin{equation}\label{eq:twisted2coboundary}
(\delta\epsilon)_{h}(g_1,g_2) = \frac{\epsilon_{\bar g_1 h g_1}(g_2)\epsilon_h(g_1)}{\epsilon_h(g_1g_2)}
\end{equation}
for some $\epsilon: G\to U(1)[G]$. 
\end{example}

\subsubsection{Slant product}

Given a 3-cocycle $\omega$ one obtains a twisted 2-cocycle $\tau(\omega)$ defined by
\begin{equation} \label{eq:slant product}
    \tau(\omega)_g(k, l) = \frac{  \omega(g, k, l) \, \omega(k, l, \bar l \, \bar k g k l)  }{ \omega(k, \bar k g k, l)   }
\end{equation}
for all $g, k, l \in G$. This lifts to a well defined group homomorphism $\tau : H^3(G, U(1)) \rightarrow H^2(G, U(1)[G])$ called the \emph{slant product} (also called the \emph{loop transgression}). 
%%%%%%%%%%%%%%%%%%%%%%%%%%%%%%%%%%%%%%%%%%%%%%%%%%%
%%%%%%%%%%%%%%%%%%%%%%%%%%%%%%%%%%%%%%%%%%%%%%%%%%%
%%%%%%%%%%%%%%%%%%%%%%%%%%%%%%%%%%%%%%%%%%%%%%%%%%%
\section{Proof of Proposition \ref{prop:anomaly}} \label{app:proof of anomaly proposition}

That the phases $\omega_j$ form a 3-cocycle, as well as the fact that the class $[\omega_j]$ is independent of the choice of right restriction and the choice of fusion operators is well known, see for example~\cite[Appendix B]{else2014classifying} for proofs. Since a right restriction at $j'$ can be viewed as a right restriction at $j$ with perhaps a different defect size, this also shows independence from $j$. It follows that $\ano(\al)$ is well defined.

Let us now prove items \ref{propitem:anomaly decoupled implies trivial} through \ref{propitem:anomaly constant on stable equivalence classes}.
\begin{enumerate}
    \item If $\al$ is decoupled then we can take a right restriction $\al_{\geq}$ such that $g \mapsto \al_{\geq}^{(g)}$ is a group homomorphism. The associated fusion operators can all be taken to be the identity so the associated 3-cocycle is identically one. The anomaly is therefore trivial.

    \item Choose a cut $j \in \Z$ so that $[j-2R, j + 6R) \subset I$. From the proof of \cite[~Theorem 9, (4)]{Gross_2012} we find that there are unitaries $U_g \in \caA_{[j-R, j+R)}$ such that both $\Ad[U_g] \circ \al^{(g)} = \al_{<}^{(g)} \al_{\geq}^{(g)}$ and $\Ad[U_g] \circ \beta^{(g)} = \beta_{<}^{(g)} \beta_{\geq}^{(g)}$ factor across the cut, \ie $\al_{<}$ and $\beta_{<}$ are supported on $(-\infty, j)$ while $\al_{\geq}$ and $\beta_{\geq}$ are supported on $[j, \infty)$. Then $\al_{\geq}$ and $\beta_{\geq}$ are right restrictions at $j$ of defect size $2R$ of the LPSs $\al$ and $\beta$ respectively, and $\al_{\geq}|_{\caA_I} =      \beta_{\geq}|_{\caA_I}$. Since $\al_{\geq}$ and $\beta_{\geq}$ agree on $I$ and the fusion operators of these right restrictions are supported on $[j - 2R, j + 3R)$ we find that the automorphisms
    \begin{align*}
        \Ad[\Fusion_{\al}(g, h)] &= \al_{\geq}^{(g)} \circ \al_{\geq j}^{(h)} \circ \big( \al_{\geq}^{(gh)} \big)^{-1} \\
        \Ad[\Fusion_{\beta}(g, h)] &= \beta_{\geq}^{(g)} \circ \beta_{\geq}^{(h)} \circ \big( \beta_{\geq}^{(gh)} \big)^{-1}
    \end{align*}
    are identical. We can therefore choose the same fusion operators for $\al_{\geq}$ and $\beta_{\geq}$ and so $\Omega(\al) = \Omega(\beta)$.

    \item Let $\al_{\geq}$ and $\beta_{\geq}$ be right restrictions at some $j \in \Z$ of $\al$ and $\beta$ respectively. Let $\Fusion_{\al}(g, h)$ and $\Fusion_{\beta}(g, h)$ be associated fusion operators and $\omega_{\al}$ and $\omega_{\beta}$ the corresponding 3-cocycles. Then $\al_{\geq} \otimes \beta_{\geq}$ is a right restriction of $\al \otimes \beta$ with associated fusion operators $\Fusion(g, h) = \Fusion_{\al}(g, h) \otimes \Fusion_{\beta}(g, h)$ and corresponding 3-cocycle $\omega = \omega_{\al} \cdot \omega_{\beta}$. Therefore $\ano(\al \otimes \beta) = [\omega] = [\omega_{\al} \cdot \omega_{\beta}] = [\omega_{\al}] \cdot [\omega_{\beta}] = \ano(\al) \cdot \ano(\beta)$, as required.
    
    \item Suppose $\al \sim_0 \beta$ are symmetries whose ranges are bounded by $R$, defined on the same spin chain $\caA$. Then there is a FDQC $\gamma$ such that $\beta = \gamma^{-1} \circ \al \circ \gamma$. Since $\gamma$ is a FDQC there is a $C > 0$ and a decomposition $\gamma = \gamma_L \circ \gamma_R$ of $\gamma$ into FDQCs $\gamma_L$ and $\gamma_R$ such that $\gamma_L$ acts as identity on $\caA_{\geq C}$ and $\gamma_R$ acts as identity on $\caA_{\leq -C}$. Then
    $$ \al \sim_0 \gamma_L^{-1} \circ \al \circ \gamma_L \sim_0 \gamma_R^{-1} \circ \gamma_L^{-1} \circ \al \circ \gamma_L \circ \gamma_R = \beta.$$
    But $\al$ and $\gamma_L^{-1} \circ \al \circ \gamma_L$ agree on $\caA_{\geq (C+R)}$ so by local computability $\ano(\al) = \ano( \gamma_L^{-1} \circ \al \circ \gamma_L )$. Similarly $\gamma_L^{-1} \circ \al \circ \gamma_L$ and $\beta$ agree on $\caA_{\leq -(C+R)}$ so by local computability $\ano( \gamma_L^{-1} \circ \al \circ \gamma_L) = \ano(\beta)$, yielding $\ano(\al) = \ano(\beta)$.

    If $\al'$ is a decoupled symmetry then $\ano(\al') = [1]$ by item \ref{propitem:anomaly decoupled implies trivial} and $\ano(\al \otimes \al') = \ano(\al) \cdot \ano(\al') = \ano(\al)$ by item \ref{propitem:anomaly multiplicative}. Together with the invariance of the anomaly under $\sim_0$, this shows that the anomaly is constant on stable equivalence classes. \hfill$\qedsymbol$\vskip.5cm
\end{enumerate}
%%%%%%%%%%%%%%%%%%%%%%%%%%%%%%%%%%%%%%%%%%%%%%%%%%%
%%%%%%%%%%%%%%%%%%%%%%%%%%%%%%%%%%%%%%%%%%%%%%%%%%%
%%%%%%%%%%%%%%%%%%%%%%%%%%%%%%%%%%%%%%%%%%%%%%%%%%%
\section{Proof of Proposition \ref{prop:obstruction to covariant right restriction}} \label{app:proof of proposition covariant obstruction}

We first show that the class $[\lambda] \in H^2(G, U(1)[G])$ is independent of the choice of right restriction. Let $\al_{\geq}$ and $\tilde \al_{\geq}$ be right restriction of the symmetry $\al$. Then there are local unitaries $\{W_g\}_{g\in G}$ such that
$$
\tilde\alpha_{\ge}^{(g)} := \Ad[W_g] \circ  \alpha_{\ge}^{(g)}. 
$$
If $\Crossing_g(k)$ are crossing operators associated to $\al_{\geq}$ then crossing operators $\tilde \Crossing_g(k)$ associated to $\tilde \al_{\geq}$ must satisfy
$$
\Ad[\tilde \Crossing_g(h)] = \alpha^{(h)} \circ \tilde \alpha_{\ge}^{(\bar h g h)} \circ \left(\alpha^{(h)}\right)^{-1} \circ \left(\tilde\alpha_{\ge}^{(g)}\right)^{-1} = \Ad \left[\alpha^{(h)}(W_{\bar h g h}) \, \Crossing_g(h) \, W_g^* \right].
$$
Therefore
\begin{equation}
    \tilde \Crossing_g(h) = \varepsilon_g(h) \alpha^{(h)}\big(W_{\bar h g h}\big) \, \Crossing_g(h) \, W_g^*
\end{equation}
for some phase map  $\varepsilon: G \to U(1)[G]$. Let $\lambda$ be the twisted 2-cocycle corresponding to the crossing operators $\Crossing_g(k)$. By a straightforward computation we find that the twisted 2-cocycle $\tilde\lambda$ corresponding to the $\tilde \Crossing_g(k)$ is
$$
\tilde \lambda_g(k,l) = \dfrac{\varepsilon_{\bar kgk}(l) \varepsilon_g(k)}{\varepsilon_g(kl)} \lambda_g(k,l).
$$
They differ up to a twisted 2-coboundary, conform \eqref{eq:twisted2coboundary}, so $[\tilde \lambda] = [\lambda]$\footnote{
This construction immediately yields a way to turn $\lambda$'s which are coboundaries to 1, by setting $\tilde{\Crossing}(g,h) = \epsilon_g(h) \Crossing_g(h)$.
}. This shows that $\Lambda(\al)$ is well defined.

The proofs of items \ref{propitem:obstruction decoupled implies trivial} through \ref{propitem:obstruction constant on stable equivalence classes} are virtually identical to the corresponding proofs of items \ref{propitem:anomaly decoupled implies trivial} through \ref{propitem:anomaly constant on stable equivalence classes} in Appendix \ref{app:proof of anomaly proposition}. We do not repeat the details here.

%%%%%%%%%%%%%%%%%%%%%%%%%%%%%%%%%%%%%%%%%%%%%%%%%%%
%%%%%%%%%%%%%%%%%%%%%%%%%%%%%%%%%%%%%%%%%%%%%%%%%%%
%%%%%%%%%%%%%%%%%%%%%%%%%%%%%%%%%%%%%%%%%%%%%%%%%%%
\section{Computation of obstructions to covariant right restrictions} \label{app:general_covariance}

Let $\al$ be the symmetry with anomaly $[\omega]$ constructed in Section \ref{sec:examples}. Let $\al_{\geq}$ be the right restriction at $a \in \Z$ consisting of all gates defining $\al$ that are supported in $[a, \infty)$. By decomposing
$$
\al^{(k)} = \al_{\leq}^{(k)} \circ \al_{\geq}^{(k)} \circ \Ad[V^{(k)}_{a-1, a}]
$$
we find that the associated crossing operators satisfy
\begin{align*}
    \Ad[ \Crossing_g(k) ] &= \al^{(k)} \circ \al_{\geq}^{(\bar k g k)} \circ \al^{(\bar k)} \circ \big( \al_{\geq}^{(g)} \big)^{-1} \\
    &= \Ad[L_{a-1}^{(k)}] \circ \al_{\geq}^{(k)} \circ \Ad[V_{a-1, a}^{(k)}] \circ \al_{\geq}^{(\bar k g k)} \circ \Ad[V_{a-1, a}^{(k)}]^* \circ \big( \al_{\geq}^{(g)} \circ \al_{\geq}^{(k)} \big)^{-1} \circ \Ad[L_{a-1}^{(k)}]^* \\
    &= \Ad[L_{a-1}^{(k)}] \circ \al_{\geq}^{(k)} \circ \left( \Ad[V_{a-1, a}^{(k)}] \circ \al_{\geq}^{(\bar k g k)} \circ \Ad[V_{a-1, a}^{(k)}]^* \circ \big(\al_{\geq}^{(\bar k g k)} \big)^{-1} \right) \\
    &\quad\quad\quad \circ \al_{\geq}^{(\bar k g k)} \circ \big( \Ad[\Fusion_a(g, k)] \circ \al_{\geq}^{(gk)} \big)^{-1} \circ \Ad[L_{a-1}^{(k)}]^*
\end{align*}
where $\Fusion_a$ are the fusion operators associated to $\al_{\geq}$ as in Lemma \ref{lem:example fusion operators}. The commutator expression is given by
$$
\Ad[V_{a-1, a}^{(k)}] \circ \al_{\geq}^{(\bar k g k)} \circ \Ad[V_{a-1, a}^{(k)}]^* \circ \big(\al_{\geq}^{(\bar k g k)} \big)^{-1} = \Ad[W_a(g, k)]
$$
where the unitary $W_a(g, k)$ is supported on $\{a-1, a\}$ and is diagonal in the group basis:
$$
W_a(g, k) | g_{a-1}, g_a \rangle = \frac{\omega(k, g_a, \bar g_a g_{a-1})}{\omega(k, \bar k \bar g k g_a, \bar g_a \bar k g k g_{a-1})} \, | g_{a-1}, g_a \rangle.
$$
Commuting $W_a(g, k)$ with $\Ad[L_{a-1}^{(k)}] \circ \al_{\geq}^{(k)}$ we obtain
\begin{align*}
\Ad[\Crossing_g(k)] &= \Ad[W'_a(g, k)] \circ \al_{\geq}^{(k)} \circ \al_{\geq}^{(\bar k g k)} \circ \big( \al_{\geq}^{(gk)} \big)^{-1} \circ \Ad[\Fusion_a(g, k)^*] \\
&= \Ad[W'_a(g, k)] \circ \Ad[ \Fusion_a(k, \bar k g k) ] 
\circ \Ad[\Fusion_a(g,k)^*] \\
&= \Ad \big[  W'_a(g, k) \Fusion_a(k, \bar k g k) \Fusion_a(g, k)^* \big]
\end{align*}
with
$$
W'_a(g, k) | g_{a-1}, g_a \rangle = \frac{\omega(k, \bar k g_a, \bar g_a g_{a-1})}{\omega(k, \bar k \bar g g_a, \bar g_a g g_{a-1})} \, | g_{a-1}, g_a \rangle.
$$
We can therefore take $\Crossing_g(k) = W'_a(g, k) \Fusion_a(k, \bar k g k) \Fusion_a(g, k)^*$, which is supported on $\{a-1, a\}$.

Let us now compute the action of $\al^{(k)} \big( \Crossing_{\bar k g k}(l) \big) \, \Crossing_g(k) \Crossing_g(kl)^*$ on a product state $| (g_i) \rangle$ in the group basis. Noting that $\al^{(k)}$ makes $\Crossing_{\bar k g k}(l)$ act on the product state $|(\bar k g_i) \rangle$, we find 
\begin{align*}
    \al^{(k)} &\big( \Crossing_{\bar k g k}(l) \big) \, \Crossing_g(k) \Crossing_g(kl)^* | (g_i) \rangle = \left( \frac{  \omega(l, \bar l \, \bar k g_a, \bar g_a g_{a-1}) \, \omega(l, \bar l \, \bar k g k l, \bar l \, \bar k \, \bar g g_a)  }{  \omega(l, \bar l \, \bar k \, \bar g g_a, \bar g_a g g_{a-1}) \, \omega(\bar k g k, l , \bar l \, \bar k \, \bar g g_a)  } \right) \\
    &\times \left(   \frac{  \omega(k, \bar k g_a, \bar g_a g_{a-1}) \, \omega( k, \bar k g k, \bar k \, \bar g g_a )  }{  \omega(k, \bar k \, \bar g g_a, \bar g_a g g_{a-1}) \, \omega( g, k, \bar k \, \bar g \, g_a )  }  \right) 
    \times \left(  \frac{  \omega( kl, \bar l \, \bar k g_a, \bar g_a g_{a-1} ) \, \omega(kl, \bar l \, \bar k g k l, \bar l \, \bar k \, \bar g g_a )  }{  \omega(kl, \bar l \, \bar k \, \bar g g_a, \bar g_a g g_{a-1}) \, \omega( g, kl, \bar l \, \bar k \, \bar g g_a )  }  \right)^{-1}  \, |(g_i) \rangle \\
    \intertext{Successively applying 3-cocycle relations for elements $(k, l, \bar l \, \bar k g_a, \bar g_a g_{a-1})$, then $(k, l, \bar l \, \bar k \, \bar g g_a, \bar g_a g g_{a-1})$, then $(k, l, \bar l \, \bar k g k l, \bar l \, \bar k \, \bar g g_a)$, then $(k, \bar k g k, l, \bar l \, \bar k \, \bar g g_a)$, and finally $(g, k, l, \bar l \, \bar k \, \bar g \, g_a)$ this becomes}
    &= \frac{  \omega(k, \bar k g k, l)  }{ \omega(g, k, l) \, \omega(k, l, \bar l \, \bar k g k l)   } | (g_i) \rangle = \tau(\omega)_g(k, l)^{-1} \, | (g_i) \rangle
\end{align*}
where $\tau(\omega)$ is the \emph{slant product} of $\omega$, see Eq. \eqref{eq:slant product}. This shows that
\begin{proposition} \label{prop:obstruction to covariant right restriction as function of anomaly}
    For any symmetry $\al$ the obstruction $\Lambda(\al)$ to covariant right restrictions is a function of its anomaly $\ano(\al) = [\omega]$ given by
    $$
    \Lambda(\al) = [ \tau(\omega) ]^{-1}
    $$
    where $\tau$ is the slant product, see Eq. \eqref{eq:slant product}.
\end{proposition}

\begin{proof}
    We verified the claimed equality for the example symmetries $\al_{\omega}$ with arbitrary anomaly $[\omega]$ constructed in Section \ref{sec:examples}. If $\al$ is an arbitrary symmetry with anomaly $[\omega]$ then by Theorem \ref{thm:classification} it is stably equivalent to $\al_{\omega}$. By Proposition \ref{prop:obstruction to covariant right restriction} the obstruction $\Lambda$ is constant on stable equivalence classes. We conclude that $\Lambda(\al) = \Lambda(\al_{\omega}) = [ \tau(\omega) ]^{-1}$.
\end{proof}

\section*{Data Availability}
Data sharing is not applicable to this article as no new data were created or analysed in this study.

 \section*{Conflict of interest}
 The authors declare no conflict of interest.

%% BIBLIOGRAPHY
\bibliographystyle{unsrturl}
\bibliography{bib}

\end{document}